\documentclass[a4paper,onecolumn,11pt,unpublished]{quantumarticle}
\pdfoutput=1

\usepackage[utf8]{inputenc}
\usepackage{amsmath,amsthm}
\usepackage{amssymb}
\usepackage{bbold}
\usepackage{xparse}
\usepackage{physics}
\usepackage{hyperref}

\usepackage{tikz}
\usepackage{tikz-cd}
\usepackage{tikzit}
\input{styles.tikzdefs}

\tikzstyle{white}=[fill=white, draw=black, shape=circle]
\tikzstyle{thickwhite}=[fill=white,very thick, draw=black, shape=circle]
\tikzstyle{black}=[fill=black, draw=black, shape=circle]
\tikzstyle{sqr}=[fill=white, draw=black, shape=rectangle]
\tikzstyle{wgrey}=[fill=white, draw={rgb,255: red,191; green,191; blue,191}, shape=circle]
\tikzstyle{bgrey}=[fill={rgb,255: red,191; green,191; blue,191}, draw={rgb,255: red,191; green,191; blue,191}, shape=circle]

\tikzstyle{discarding}=[fill=white, draw=black, shape=circle, style=upground, scale = 1.25]
\tikzstyle{smalldiscarding}=[fill=white, draw=black, style=upground, scale=0.75]
\tikzstyle{backdiscard}=[fill=white, draw=black, shape=circle, style=downground, scale=0.5]
\tikzstyle{smallbackdiscard}=[fill=white, draw=black, shape=circle, style=downground, scale=0.5]
\tikzstyle{state}=[fill=white, draw=black, style=triang, tikzit shape=rectangle]
\tikzstyle{kstate}=[fill=white, draw=black, style=kpoint, tikzit shape=rectangle]
\tikzstyle{kstateconj}=[fill=white, draw=black, style=kpoint conjugate, tikzit shape=rectangle]
\tikzstyle{kstateBIG}=[fill=white, draw=black, style=big kpoint, tikzit shape=rectangle]
\tikzstyle{effect}=[fill=white, draw=black, style=triangdag]
\tikzstyle{keffect}=[fill=white, draw=black, style=kpoint adjoint]
\tikzstyle{keffectconj}=[fill=white, draw=black, style=kpoint transpose]
\tikzstyle{morphdag}=[style=mapdag]
\tikzstyle{morph}=[style=hadamard]
\tikzstyle{WIDEmorph}=[style=hadamard, minimum width=14mm]
\tikzstyle{morphtrans}=[style=maptrans]
\tikzstyle{morphconj}=[style=mapconj]
\tikzstyle{CPMmorph}=[style=dmap]
\tikzstyle{CPMmorphconj}=[style=dmapconj]
\tikzstyle{CPMmorphdag}=[style=dmapdag]
\tikzstyle{CPMmorphtrans}=[style=dmaptrans]
\tikzstyle{CPMstate}=[fill=white, draw=black, style=triang, doubled]
\tikzstyle{CPMstateBIG}=[fill=white, draw=black, style={triang_lesssep}, doubled]
\tikzstyle{CPMkstate}=[fill=white, draw=black, style=kpoint, tikzit shape=rectangle, doubled]
\tikzstyle{CPMkstateconj}=[fill=white, draw=black, style=kpoint conjugate, tikzit shape=rectangle, doubled]
\tikzstyle{CPMkstateBIG}=[fill=white, draw=black, style=big kpoint, tikzit shape=rectangle, doubled]
\tikzstyle{CPMkeffect}=[fill=white, draw=black, style=kpoint adjoint, doubled]
\tikzstyle{CPMkeffectconj}=[fill=white, draw=black, style=kpoint transpose, doubled]
\tikzstyle{UHfB}=[fill=white, draw=black, style=triangdag, doubled, inner sep=-2pt]
\tikzstyle{leak}=[style=tinypoint, regular polygon rotate=-90]
\tikzstyle{leakfill}=[style=tinypoint, regular polygon rotate=-90, fill=black]
\tikzstyle{Z}=[style=dot, fill=green]
\tikzstyle{X}=[style=dot, fill=red]
\tikzstyle{black_dot}=[style=dot, fill=black]
\tikzstyle{white_dot}=[style=dot, fill=white]
\tikzstyle{qblack_dot}=[style=ddot, fill=black]
\tikzstyle{qwhite_dot}=[style=ddot, fill=white]
\tikzstyle{whitephase}=[style=wphase dot, fill=white]
\tikzstyle{qredphase}=[style=phase dot, fill=red]
\tikzstyle{qgreenphase}=[style=phase dot, fill=green]
\tikzstyle{had}=[style=hadamard, doubled]
\tikzstyle{box}=[style=hadamard]
\tikzstyle{bigbox}=[style=hadamard, minimum height=4mm, minimum width=8mm]
\tikzstyle{classhad}=[style=hadamard]
\tikzstyle{antipode}=[style=anti]

\tikzstyle{dottededge}=[-, dash pattern=on 1pt off 0.7pt]
\tikzstyle{double edge}=[-, style=doubled, draw=black, tikzit draw={rgb,255: red,18; green,168; blue,191}]
\tikzstyle{arrow}=[->]
\tikzstyle{new edge style 1}=[-, draw={rgb,255: red,242; green,233; blue,206}, fill={rgb,255: red,242; green,233; blue,206}]
\tikzstyle{morphism_shade}=[-, draw=black, fill={rgb,255: red,242; green,233; blue,206}, line join=bevel]
\tikzstyle{supermap_shade}=[-, fill={rgb,255: red,216; green,215; blue,242}, draw=black, line join=bevel]
\tikzstyle{hole_shade}=[-, fill=white, draw=black, line join=bevel]
\tikzstyle{new edge style 2}=[-, draw={rgb,255: red,14; green,188; blue,83}]
\tikzstyle{new edge style 3}=[<-, draw={rgb,255: red,234; green,209; blue,255}]
\tikzstyle{new edge style 4}=[<-, draw={rgb,255: red,0; green,106; blue,106}]
\tikzstyle{new edge style 5}=[-, draw={rgb,255: red,214; green,110; blue,62}]
\tikzstyle{new edge style 6}=[-, draw={rgb,255: red,174; green,20; blue,174}]
\tikzstyle{new edge style 0}=[-, fill=none, draw={rgb,255: red,0; green,106; blue,106}]

\usepackage{graphicx}   
\usepackage{subcaption}

\newtheorem{theorem}{Theorem}[section]
\newtheorem{remark}{Remark}[section]
\newtheorem{corollary}{Corollary}[section]
\newtheorem{definition}{Definition}[section]
\newtheorem{lemma}{Lemma}[section]
\newtheorem{proposition}{Proposition}[section]

\newcommand{\ca}{\mathcal A}
\newcommand{\cb}{\mathcal B}

\newcommand{\ce}{\mathcal E}
\newcommand{\cf}{\mathcal F}

\newcommand{\ch}{\mathcal H}

\newcommand{\ck}{\mathcal K}
\newcommand{\cl}{\mathcal L}

\newcommand{\cz}{\mathcal Z}

\newcommand{\tp}{\tilde{P}}

\newcommand{\chlc}{\mathcal{H}_{\LL}^{\chi}}
\newcommand{\chrc}{\mathcal{H}_{\R}^{\chi}}
\newcommand{\chlz}{\mathcal{H}_{\LL}^{\zeta}}
\newcommand{\chrz}{\mathcal{H}_{\R}^{\zeta}}
\newcommand{\chlx}{\mathcal{H}_{\LL}^{\xi}}
\newcommand{\chrx}{\mathcal{H}_{\R}^{\xi}}

\newcommand{\pic}{\pi^{\chi}}
\newcommand{\sigc}{\sigma^{\chi}}

\def\be{\begin{equation}}
\def\ee{\end{equation}}
\def\ba{\begin{align}}
\def\ea{\end{align}}

\newcommand{\mub}{{\Bar{\mu}}}

\DeclareMathOperator{\Span}{Span}
\DeclareMathOperator{\vnalg}{vnalg}
\DeclareMathOperator{\spl}{split}
\DeclareMathOperator{\lean}{lean}
\DeclareMathOperator{\balanced}{balanced}
\DeclareMathOperator{\loc}{loc}
\DeclareMathOperator{\stloc}{stloc}
\DeclareMathOperator{\cons}{cons}
\DeclareMathOperator{\LL}{L}
\DeclareMathOperator{\R}{R}
\DeclareMathOperator{\MM}{M}

\DeclareMathOperator{\Atproj}{AtomProj}

\newcommand{\tikzcircle}[2][black,fill=black]{\tikz[baseline=-0.5ex]\draw[#1,radius=#2] (0,0) circle ;}%

\newcommand{\id}{\mathbb{1}}

\begin{document}

\title{Picturing general quantum subsystems}

\author{Octave Mestoudjian}
\affiliation{Université Paris-Saclay, Inria, CNRS, LMF, 91190 Gif-sur-Yvette, France}

\author{Matt Wilson}
\affiliation{Université Paris-Saclay, CentraleSupélec, Inria, CNRS, LMF, 91190 Gif-sur-Yvette, France}
\affiliation{Programming Principles Logic and Verification Group, University College London, United Kingdom}

\author{Augustin Vanrietvelde}
\affiliation{Télécom Paris, Institut Polytechnique de Paris, Inria Saclay, Palaiseau, France}

\author{Pablo Arrighi}
\affiliation{Université Paris-Saclay, Inria, CNRS, LMF, 91190 Gif-sur-Yvette, France}

\maketitle

\begin{abstract}
We extend the usual process-theoretic view on locality and causality in subsystems (based on the tensor product case) to general quantum systems (i.e.\ possibly non-factor, finite-dimensional von Neumann algebras). To do so, we introduce a primitive notion of splitting maps within dagger symmetric monoidal categories. Splitting maps give rise to subsystems that admit comparison via a preorder called comprehension, and support an adaptation of the usual categorical trace. We show that the comprehension preorder precisely captures the inclusion partial order between von Neumann algebras, and that the splitting map trace captures the natural notion of von Neumann algebra trace. As a consequence of the development of these diagrammatic tools, we prove that the known equivalence between semi-causality and semi-localisability for factor subsystems extends to all (including non-factor) subsystems.

\end{abstract}

\section{Introduction}

The notions of system and the rules by which they can be composed or decomposed are at the foundation of any physical theory. In quantum theory these rules challenge our preconceived notions of subsystems: a system $AB$, for instance, cannot be understood as just the sum of its parts $A$ and $B$ in the classical sense, but also describes quantum correlations between them. Traditionally, the composition of quantum systems is represented in terms of the tensor product of Hilbert spaces, whereas the decomposition is represented in terms of subalgebras of the algebra of operators on the Hilbert space at hand.

When considering that these algebras should be factor, the two approaches coincide and are represented using the quantum circuit model \cite{deutsch1989quantum}, which is standardly used in quantum computation \cite{Nielsen_Chuang_2010}, quantum information theory \cite{Wilde_2017} and quantum foundations \cite{chiribella_purification, gogioso_cpt}. From the physical perspective, factor algebras and their associated circuit representation equip quantum theory with the fundamental notions of subsystem, state-reduction, locality, and causality, used in quantum foundations topics such as quantum reference frames \cite{Bartlett2007, Giacomini2017, ahmad2022, relativeSubsystems}, reconstruction programs \cite{chiribella_purification, Selby2021reconstructing}, indefinite causal structures \cite{PhysRevA.88.022318, oreshkov_causal_order}, quantum cellular automata \cite{FeynmanQCA, schumacher2004, Arrighi2007}, etc.

However, many situations unravel the limiting aspects of this standard definition of subsystem, even in finite dimensional quantum theory. Some quantum evolutions have a causal structure that can only be made visible when mixing together the tensor product and direct sum \cite{lorenz_unitary, vanrietvelde2023consistentcircuitsindefinitecausal, Barrett2021}. Many physical state spaces satisfy some superselection rules that are not respected by the tensor product composition, the most common example being the parity conservation for fermions \cite{dariano2014, Eon2023relativistic}. In quantum gravity, where superpositions of geometries (of space-time) are often described as superposition of graphs, the neighbours of a node exist in superposition, and so cannot be factored into naive tensor products \cite{Arrighi2024quantumnetworks, Bianchi2023}. 

A more general definition of subsystem in quantum theory, and one that solves the aforementioned issues, is to consider that a subalgebra representing a subsystem does not have to be a factor. The treatment of general quantum subsystems such as these in terms of compositional or diagrammatic principles has received some attention, for instance via representation of algebras as monoids with properties and the associated CP* construction \cite{Vicary_2010,coecke_cp,Coecke_2018}, the development of routed quantum circuits \cite{Vanrietvelde_2021, vanrietvelde2023consistentcircuitsindefinitecausal} and its categorical generalisation \cite{wilson2021composableconstraints},  the use of graphical languages for $2$-categories \cite{Claeys_2024,allen2024cpinftybeyond2categoricaldilation,Reutter_2019}, and in the many-worlds calculus \cite{Chardonnet_2025}. The understanding of general subsystems is however less well developed than it is for the factor case: in particular, somewhat suprisingly, there is no pictorial representation in the literature of the basic notion of splitting a system into parts, and of the natural pictorial notions of locality and causality that ought to come with such a representation. The formalisation of such a concept would likely give a natural route to extending the reasoning tools for tensor products in categorical quantum mechanics \cite{coecke_kissinger_2017, heunen_categories, abramsky_coecke} to the non-factor setting, and further a natural route to the generalisation of non-factor subsystems to theories beyond quantum, such as those that fit within the frameworks of generalised \cite{barrett_gpts}, operational \cite{chiribella_purification}, or categorical probabilistic theories \cite{gogioso_cpt}. A satisfactory treatment of the notion, would likely contribute to the attempt to gain a clearer understanding of what exactly it means, formally, theory-independently, and compositionally, to be a subsystem \cite{Kramer_2018,Chiribella_2018, gogioso_church}. 

In this article, we find a simple set of diagrammatic tools for picturing and reasoning about locality and causality in general non-factor systems. To do so, we leverage the representation of factors in the circuit model, by defining formal internalised splitting maps which consist of embeddings of a base system into a bigger system that is equipped with a standard factor decomposition. These splitting maps allow us to derive, through diagrammatic definitions and from the factorisation of the bigger space, non-factor decompositions of the initial space; and to develop associated notions of locality and inclusion between subsystems. In order to demonstrate that this purely pictorial theory of subsystems recovers the algebraic approach when applied to quantum theory, we establish an equivalence between subsystems as diagrammatic splitting maps in quantum theory and subsystems as von Neumann subalgebras, as they are more usually conceptualised \cite{viola2001, zanardi2001, zanardi2003, Chiribella_2018, vanrietvelde2025partitionsquantumtheory}.


We show how these tools can be used to prove basic structural theorems about locality and causality. A fundamental result on the nature of causality in quantum theory is the equivalence,  for a bipartite quantum channel between factor subsystems, of non-signalling from one system to the other and semi-localisability \cite{Eggeling_2002, beckman_localisable, schumacher_locality}, which demonstrates that information-theoretic non-signalling can be framed purely at the compositional level. We use the diagrammatic tools that we developed to prove that this equivalence between semi-causality and semi-localisability lifts from factor systems to arbitrary von Neumann algebras, proving as a result the equivalence for general quantum systems which naturally arise for instance in superpositions of geometries.


The plan of the paper is as follows. In the second section, we explain how subsystems can naturally be seen as von Neumann algebras and provide a few results on the structure of such subsystems. In the third section, we introduce the formalism of splitting maps that allows us to talk about diagrammatic subsystems. In the fourth section we show that the formalism of splitting maps is strictly equivalent to the one of subsystems as von Neumann algebras. And in the fifth section, we define the notion of trace of a splitting map, show that it is equivalent to the usual trace over a von Neumann algebra and use the formalism of splitting maps to provide a causal decomposition for maps that are non-signalling between general subsystems.





\section{General Quantum Subsystems as von Neumann algebras}


All the Hilbert spaces and algebras considered in this article are finite-dimensional. We will now remind a few basic facts on the structure of these algebras \cite{farenick} that represent quantum systems.

\begin{definition}[* algebra]
A (finite dimensional) {\em * algebra} $\ca$ is a finite-dimensional complex algebra equipped with an involution $\dagger$ such that:
\begin{itemize}
\item for all $a,b$ in $\ca$, $(a+b)^{\dagger} = a^{\dagger} + b^{\dagger}$
\item for all $a,b$ in $\ca$, $(ab)^{\dagger} =   b^{\dagger} a^{\dagger}$
\item for all $a$ in $\ca$ and $\lambda$ in $\mathbb{C}$, $(\lambda a)^{\dagger} = \overline{\lambda}a^{\dagger}$
\end{itemize}
A {\em* subalgebra} $\cb$ of $\ca$ is a subalgebra of $\ca$ equipped with the involution $\dagger$ of $\ca$ and closed under it. This is the notion of subalgebra that we will use in the rest of this article and we will write it $\cb \subseteq \ca$. Moreover a subalgebra $\cb \subseteq \ca$ is said to be {\em unital} if $\id_{\cb} = \id_{\ca}$, i.e. $\cb$ has the same unit element as $\ca$.
\end{definition}

\begin{definition}[von Neumann Algebra]
A (finite dimensional) {\em von Neumann algebra} is a unital * subalgebra of some $\cl(\ch)$ with, as involution, the usual adjoint coming from the inner product on $\ch$. Let us call $\vnalg(\ch)$ the set of von Neumann algebras $\ca \subseteq \cl(\ch)$.
\end{definition}

Note that in finite dimension, a von Neumann algebra is the same thing as a C* algebra. Let us give two examples of von Neumann algebras acting on $\ch$, coming from a decomposition of the space either as a tensor product or as a direct sum. Suppose that $\ch = \ch_{\LL} \otimes \ch_{\R}$, then one can define the von Neumann algebra of operators  acting on the left part of the tensor $\ca_{\otimes} = \cl(\ch_{\LL}) \otimes \id_{\ch_{\R}}$. Similarly suppose that $\ch = \ch_1 \oplus \ch_2$, then one can define the von Neumann algebra $\ca_{\oplus} = \cl(\ch_1) \oplus \mathbb{C} \id_{\ch_2}$ of operators that leave invariant $\ch_1$ and $\ch_2$ and moreover can only act as homotheties on $\ch_2$.

\begin{definition}[Commutant and centre]
Let $\ca \subseteq \cl(\ch)$ be a von Neumann algebra. The {\em commutant} of $\ca$ (within $\cl(\ch)$) is the algebra $\ca' := \{b \in \cl(\ch) | \forall a \in \ca, ab = ba \}$ of operators on $\ch$ that commute with every element of $\ca$. The centre of $\ca$ is the commutative algebra $\cz(\ca) := \ca \cap \ca'$ of the elements of $\ca$ that commute with all elements of $\ca$. Observe that due to von Neumann's bicommutant theorem \cite{takesaki2001theory}, $\cz(\ca) =  \ca \cap \ca' = \cz(\ca')$. We say that $\ca$ is {\em factor} when $\cz(\ca) = \mathbb{C} \id_{\ch} $.
\end{definition}

In the case of the algebra  $\ca_{\otimes} = \cl(\ch_{\LL}) \otimes \id_{\ch_{\R}}$ defined above, its commutant is the algebra $\ca_{\otimes}' = \id_{\ch_{\LL}} \otimes \cl(\ch_{\R})$ and its centre is simply the commutative algebra $\cz(\ca_{\otimes}) = \mathbb{C} \id_{\ch}$ (the algebra $\ca_{\otimes}$ is factor). In the case of the algebra $\ca_{\oplus}  = \cl(\ch_1) \oplus \mathbb{C} \id_{\ch_2}$, its commutant is $\ca_{\oplus} ' = \mathbb{C} \id_{\ch_1} \oplus \cl(\ch_2)$ which may seem close to what we get for $\ca_{\otimes}'$ but its centre is the (fundamentally different from $\cz(\ca_{\otimes})$) non-trivial $ \cz(\ca_{\oplus} ) = \mathbb{C} \id_{\ch_1} \oplus \mathbb{C} \id_{\ch_2}$.

\begin{theorem}[Atomic projectors] \cite{vanrietvelde2025partitionsquantumtheory}
Let $\cz$ be a commutative (finite dimensional) von Neumann algebra. Then there exist a unique family $\{ \pi_i \}$ of non-null, orthogonal (i.e. $\forall i, \pi_i^{\dagger} = \pi_i$), and pairwise orthogonal (i.e. $\forall i,j, \pi_i \pi_j = \delta_{ij} \pi_i$) projectors such that $\cz = \Span(\{ \pi_i \})$. They are called the {\em atomic projectors} of $\cz$ and we will write $\Atproj(\cz) = \{ \pi_i \}$.
\end{theorem}

Let us now give a representation theorem for von Neumann algebras which will hopefully make these definitions more concrete in the general case.

\begin{theorem}[Artin-Wedderburn] \label{AW} \cite{farenick}
Let $\ca \subseteq \cl(\ch)$ be a (finite dimensional) von Neumann algebra. Then there exists Hilbert spaces $(\ch_{\LL}^{i})_{i \in I}$ and $(\ch_{\R}^{i})_{i \in I}$ and a unitary map $U: \ch \rightarrow \bigoplus_i (\ch_{\LL}^{i} \otimes \ch_{\R}^{i})$ such that 
\begin{equation}
U \ca U^{\dagger} = \bigoplus_i \cl(\ch_{\LL}^{i}) \otimes \id_{\ch_{\R}^{i}} \, .
\end{equation}
We will call such a unitary a {\em representation unitary} for the algebra $\ca$.
\end{theorem}

In the case of the algebra $\ca_{\otimes}$ defined earlier in this section as acting on the left part of the space $\ch = \ch_{\LL} \otimes \ch_{\R}$, this result is completely transparent: the direct sum has a single term and the theorem simply says that $\ch = \ch_{\LL} \otimes \ch_{\R}$ and that $\ca_{\otimes} = \cl(\ch_{\LL}) \otimes \id_{\ch_{\R}}$. However in the case of $\ca_{\oplus}  = \cl(\ch_1) \oplus \mathbb{C} \id_{\ch_2}$, the theorem gives us a non-trivial decomposition. Indeed, it tells us that the underlying space $\ch$ is (unitarily) isomorphic to $(\ch_1 \otimes \mathbb{C}_1) \oplus (\mathbb{C}_2 \otimes \ch_2)$ and that, following the same mapping,  $\ca_{\oplus} = \cl(\ch_1) \oplus \mathbb{C} \id_{\ch_2} \cong (\cl(\ch_1) \otimes \id_{\mathbb{C}_1}) \oplus (\cl(\mathbb{C}_2) \otimes \id_{\ch_2})$. 

\begin{remark}
Note that the same representation unitary also gives a representation for the algebra $\ca'$ as 
\begin{equation}
U \ca' U^{\dagger} = \bigoplus_i  \id_{\ch_{\LL}^{i}} \otimes \cl(\ch_{\R}^{i}) \, .
\end{equation}
\end{remark}

Remark also that the Artin-Wedderburn theorem makes the atomic projectors appear very naturally. Indeed, if $\ca \subseteq \cl(\ch)$ is a von Neumann algebra and if the Artin-Wedderburn theorem applied to $\ca$ is stated as 
\begin{equation}
U \ca U^{\dagger} = \bigoplus_i \cl(\ch_{\LL}^{i}) \otimes \id_{\ch_{\R}^{i}} \, ,
\end{equation}
then the atomic projectors of $\cz(\ca)$ (and thus of $\ca$ by extension) are precisely the orthogonal projectors onto the $\ch^{i} = U^{\dagger}(\ch_{\LL}^{i} \otimes \ch_{\R}^{i} )$.

\begin{theorem}\label{blocks are factor}
Let $\ca \subseteq \cl(\ch)$ be a (finite dimensional) von Neumann algebra. Let $\{ \pi_i \}$ be the family of the atomic projectors of its center and $\ca \cong \bigoplus_i \cl(\ch_{\LL}^{i}) \otimes \id_{\ch_{\R}^{i}}$ be a decomposition given by Artin-Wedderburn theorem. Then 

\be
\ca = \bigoplus_i \pi_i \ca \cong \bigoplus_i \cl(\ch_{\LL}^{i}) \otimes \id_{\ch_{\R}^{i}} \cong \bigoplus_i \cl(\ch_{\LL}^{i}) \, ,
\ee
meaning in particular that each of the $\pi_i \ca$ is a factor algebra whose centre is $\cz(\pi_i \ca) = \mathbb{C} \pi_i$. 
\end{theorem}

Identifying $\ch$ with $\bigoplus_i (\ch_{\LL}^{i} \otimes \ch_{\R}^{i})$, one can now give a notion of partial trace over the algebra $\cb = \ca'$. This map from $\cl(\ch)$ to $\bigoplus_i \cl(\ch_{\LL}^{i})$ that acts as an injective von Neumann algebra homomorphism on $\ca$ is essentially an extractor of the algebra $\ca$ from $\cl(\ch)$.

\begin{definition} [Partial trace over an algebra]\cite{Chiribella_2018} \label{trace over an algebra}
The {\em partial trace over the (finite dimensional) algebra} $\cb = \ca'$ is the map $\Tr_{\cb} : \cl(\ch) \longrightarrow \bigoplus_i \cl(\ch_{\LL}^{i})$ defined by: 
\begin{equation}
\Tr_{\cb}(\cdot) := \bigoplus_i \Tr_{\ch_{\R}^{i}}(\pi_i \cdot \pi_i) \, ,
\end{equation}
where the $\pi_i$ are the projectors onto the $\ch_{\LL}^{i} \otimes \ch_{\R}^{i}$.
\end{definition}

When the algebra considered is $\ca_{\otimes}$, we have seen that the decomposition given by Theorem \ref{AW} is $\ch = \ch_{\LL} \otimes \ch_{\R}$ and tracing over the algebra $\ca'$ meets the usual notion of tracing over the Hilbert space $\ch_{\R}$ of the factorisation. 

In a similar fashion, one can define the trace over the algebra $\ca_{\oplus} = \cl(\ch_1) \oplus \mathbb{C} \id_{\ch_2}$. Observe that $\cl(\ch) \cong \cl(\ch_1) \oplus \cl(\ch_1,\ch_2) \oplus \cl(\ch_2,\ch_1) \oplus \cl(\ch_2)$ and let us decompose elements of $A \in \cl(\ch)$ in this way, as $A = A_{1,1} \oplus A_{1,2} \oplus A_{2,1} \oplus A_{2,2}$. Applying Definition \ref{trace over an algebra}, we obtain that for any $A \in \cl(\ch)$, $\Tr_{\ca_{\oplus}'}(A) = A_{1,1} \oplus \Tr_{\ch_2}(A_{2,2}) \id_{\mathbb{C}_2}$, meaning that after tracing we remember about the action of $A$ from $\ch_1$ to $\ch_1$ but have lost the information about the proper action of $A$ from $\ch_2$ to $\ch_2$, which has been replaced by $\Tr_{\ch_2}(A_{2,2}) \id_{\mathbb{C}_2}$; and we have completely lost track of the action of $A$ between $\ch_1$ and $\ch_2$.

A fundamental information about a family of (sub)systems is the data of whether one of these systems is included in another one, in the sense that it is a subsystem of the latter. This information can be described by a partial order/preorder on the family of systems.

\begin{definition}[Partial order and preorder]
A {\em preorder} on a set $X$ is a binary relation, $\leq$ on $X$ that is:
\begin{itemize}
\item reflexive, i.e. $\forall x \in X, x \leq x$,
\item transitive, i.e. if $x \leq y$ and $y \leq z$ then $x \leq z$.
\end{itemize}
If $\leq$ is moreover
\begin{itemize}
\item antisymmetric, i.e. if $x \leq y$ and $y \leq x$ then $x = z$,
\end{itemize}
$\leq$ is called a {\em partial order}.
\end{definition}

\begin{proposition}
The inclusion between von Neumann algebras $\ca \subseteq \cl(\ch)$ (as it is an inclusion at the level of the sets of elements of the algebras) is a partial order.
\end{proposition}

In this paper we will be looking for a diagrammatic (as opposed to algebraic) way to reason with general quantum subsystems. In order to do so we aim to find a diagrammatic approach which is suitably equivalent to the algebraic one. The previous proposition tells us that, given a base system, the family of its algebraic subsystems forms a partially ordered set. We will see in Section \ref{comprehension} that the family of its diagrammatic subsystems forms a preordered set. The common structure being the one of preordered sets, we define, in order to compare the two approaches, two notions. The first one is the one of preorder preserving map, or monotone map.

\begin{definition}[Preorder preserving map]
Let $(X,\subseteq)$ and $(Y,\sqsubseteq)$ be preordered sets. We say that a map $f: X \longrightarrow Y$ is {\em preorder preserving} if, for all $x_1,x_2$ in $X$, $x_1 \subseteq x_2 \Rightarrow f(x_1) \sqsubseteq f(x_2)$.
\end{definition}

And the second one, because preorders can be described as a special kind of categories, is the one of equivalence of preorders, which is simply an equivalence of categories \cite{Maclane1971-MACCFT} between two preorders.

\begin{definition}[Equivalence of preorders]
 Let $(X,\subseteq)$ and $(Y,\sqsubseteq)$ be preordered sets. We say that the preorder preserving maps $F: X \longrightarrow Y$ and $G: Y \longrightarrow X$ yield an {\em equivalence of preorders} if for all $x \in X$, $x \subseteq G\circ F(x)$ and $G\circ F(x) \subseteq x$; and for all $y \in Y$, $y \sqsubseteq F\circ G(y)$ and $F\circ G(y) \sqsubseteq y$.
\end{definition}


\section{Picturing general quantum subsystems as splitting maps}\label{section3}



We will now formally define splitting maps: those mathematical objects that will allow us to split systems into a left part and a right part, as well as to define the notions of locality, strict locality, consistency and comprehension that come with this splitting. In particular, we will give a model for subsystems that relies only on the symmetric monoidal structure of the theory at hand, in contrast with the usual algebraic approach, which for instance makes direct reference to closure under sums.
In the rest of this section, $\ch$ is a fixed Hilbert space.

\subsection{Definitions}

In the standard paradigm of factor subsystems, the splitting of the system  represented by the Hilbert space $\ch$, into subsystems $Y$ and $Z$, represented by Hilbert spaces $\ch_Y$ and $\ch_Z$, amounts to defining a unitary $U : \ch \rightarrow \ch_Y \otimes \ch_Z$. The idea of splitting maps is to relax this condition and ask for the map $U$ to be an isometry instead of a unitary, which will allow us to capture more general subsystems.

\begin{definition}[Splitting map]
A {\em splitting map} $\chi$ on $\ch$ is an isometry $\chi : \ch \rightarrow \chlc \otimes \chrc$ where $\chlc$ and $\chrc$ are Hilbert spaces. We will write it diagrammatically as

\begin{equation}
  \input{figures/chi.tikz}
\end{equation}
and we will write its adjoint $\chi^{\dagger}$, which is a coisometry, as

\begin{equation}
  \input{figures/chid.tikz} \, .
\end{equation}
We will call $\spl(\ch)$ the set of splitting maps on $\ch$.
\end{definition}

Let us give a few concrete examples of splitting maps. A first example is the one that recovers the usual decomposition into subsystems and thus the factorisation of a Hilbert space. This happens when the splitting map is not only an isometry, but a unitary $\chi_{\otimes} : \ch \rightarrow \ch_{\LL} \otimes \ch_{\R}$, making the tensor product structure of the space appear. Another example of a splitting map can be given in the case of the direct sum. Indeed, if $\ch = \ch_1 \oplus \ch_2$, we have seen that Theorem \ref{AW} gives us a decomposition of the space as $\ch \cong (\ch_1 \otimes \mathbb{C}_1) \oplus (\mathbb{C}_2 \otimes \ch_2)$ which can be naturally embedded into $(\ch_1 \oplus \mathbb{C}_2) \otimes (\mathbb{C}_1 \oplus \ch_2) =  (\ch_1 \otimes \mathbb{C}_1) \oplus (\ch_1 \otimes \ch_2) \oplus (\mathbb{C}_2 \otimes \ch_2) \oplus (\mathbb{C}_2 \otimes \mathbb{C}_1)$. Naming $\ket{\emptyset_1}$ a norm-one element of $\mathbb{C}_1$ and $\ket{\emptyset_2}$ one of $\mathbb{C}_2$, we can define this embedding splitting map $\chi_{\oplus} : \ch = \ch_1 \oplus \ch_2 \hookrightarrow (\ch_1 \oplus \mathbb{C}_2) \otimes (\mathbb{C}_1 \oplus \ch_2)$ which, for any $\ket{x} = \ket{x_1} \oplus \ket{x_2} \in \ch_1 \oplus \ch_2 = \ch$, is such that $\chi_{\oplus}(\ket{x}) = \ket{x_1} \otimes \ket{\emptyset_1} + \ket{\emptyset_2} \otimes \ket{x_2}$.


\begin{remark}
Observe that $\chi^{\dagger}\chi = \id_{\ch}$ and that $\chi \chi^{\dagger}$ is an orthogonal projector of $\cl(\chlc \otimes \chrc)$ that we will call $\pi^{\chi}$ (it is the orthogonal projector on $\Im(\chi)$). Diagrammatically we get that 
\end{remark}

\begin{equation}
  \input{figures/id.tikz}
\end{equation}

and

\begin{equation}
  \input{figures/pi.tikz}
\end{equation}

This splitting into a left and a right part allows us to talk about operators on $\ch$ that are (left)-local according to $\chi$. 

\begin{definition}[Locality]
Let $\chi$ be a splitting map on $\ch$ and $A \in \cl(\ch)$. We say that $A$ is {\em $\chi$-local} if there exists $\tilde{A} \in \cl(\chlc)$ such that $A = \chi^{\dagger} (\tilde{A} \otimes \id_{\chrc} )\chi$, or in diagrammatic notation,

\begin{equation}
  \input{figures/loc.tikz} \, .
\end{equation}
We will say that $\tilde{A}$ is an {\em on-site representative} of $A$. We will call $\loc(\chi)$ the set of $\chi$-local operators of $\cl(\ch)$.
\end{definition}

Let us now illustrate this definition with the two splitting maps that we defined above. In the case of the tensor splitting $\chi_{\otimes} : \ch = \ch_{\LL} \otimes \ch_{\R} \rightarrow \ch_{\LL} \otimes \ch_{\R}$, which simply acts as the identity, the local operators are precisely the ones of the form $A = \tilde{A} \otimes \id_{\ch_{\R}}$, i.e.\ elements of $\ca_{\otimes}$, so that

\begin{equation}
\ca_{\otimes} = \loc(\chi_{\otimes}) \, .
\end{equation}

Let us now turn to the case of the direct sum and study local operators of the form  $A = \chi_{\oplus}^{\dagger} (\tilde{A} \otimes \id_{\mathbb{C}_1 \oplus \ch_2}) \chi_{\oplus}$. Observe that because $\tilde{A} \otimes \id_{\mathbb{C}_1 \oplus \ch_2}$ is trivial on the right part of the splitting, $\tilde{A} \otimes \id_{\mathbb{C}_1 \oplus \ch_2}(\ch_1 \otimes \mathbb{C}_1) \subseteq (\ch_1 \oplus \mathbb{C}_2) \otimes \mathbb{C}_1$ and $\tilde{A} \otimes \id_{\mathbb{C}_1 \oplus \ch_2}(\mathbb{C}_2 \otimes \ch_2) \subseteq (\ch_1 \oplus \mathbb{C}_2) \otimes \ch_2$. Moreover, as $\chi_{\oplus}$ is the map that isometrically embeds $\ch = \ch_1 \oplus \ch_2 \cong  (\ch_1 \otimes \mathbb{C}_1) \oplus (\mathbb{C}_2 \otimes \ch_2)$ into $(\ch_1 \oplus \mathbb{C}_2) \otimes (\mathbb{C}_1 \oplus \ch_2)$, $\chi_{\oplus}^{\dagger}$ is the map that kills the $\ch_1 \otimes  \ch_2$ and $\mathbb{C}_2 \otimes \mathbb{C}_1$ subspaces of $(\ch_1 \oplus \mathbb{C}_2) \otimes (\mathbb{C}_1 \oplus \ch_2)$ and then makes the converse identification with $\ch$. It follows that any action of $\tilde{A} \otimes \id$ which would move between subspaces $\ch_1 \otimes \mathbb{C}_1$ and $\mathbb{C}_2 \otimes \mathbb{C}_1$ or between $\mathbb{C}_2 \otimes \ch_2$ and $\ch_1 \otimes \ch_2$ will be killed by the $\chi_{\oplus}$ splitting and merging. Therefore, $A$ only depends on the diagonal actions $\tilde{A}_{1,1} \in \cl(\ch_1)$ and $\tilde{A}_{2,2} \in \cl(\mathbb{C}_2) \cong \mathbb{C}$ (which we will thus call $\alpha_2$) of $\tilde{A}$. In the end, the action of $A$ on an element $\ket{x} = \ket{x_1} \oplus \ket{x_2}$ of  $\ch$ is 
\begin{equation}
\begin{split} 
A\ket{x} & = \chi_{\oplus}^{\dagger} (\tilde{A} \otimes \id_{\mathbb{C}_1 \oplus \ch_2}) \chi_{\oplus} (\ket{x_1} \oplus \ket{x_2}) \\
& = \chi_{\oplus}^{\dagger}(\tilde{A}\ket{x_1} \otimes \ket{\emptyset_1}) \oplus \chi_{\oplus}^{\dagger}(\tilde{A}\ket{\emptyset_2} \otimes \ket{x_1}) \\
& = \chi_{\oplus}^{\dagger}(\tilde{A}_{1,1}\ket{x_1} \otimes \ket{\emptyset_1}) \oplus \chi_{\oplus}^{\dagger}(\tilde{A}_{2,2} \ket{\emptyset_2} \otimes \ket{x_1}) \\
& = \tilde{A}_{1,1}\ket{x_1} \oplus  \alpha_2 \ket{x_2} \, ,
\end{split}
\end{equation}
which directly implies that $A$ is an element of the algebra $\ca_{\oplus}$ defined earlier. Conversely, any element of $\ca_{\oplus}$ is a $\chi_{\oplus}$-local operator, by taking its action on $\ch_1$ as $\tilde{A}_{1,1}$ and its homothety ratio on $\ch_2$ as $\tilde{A}_{2,2} = \alpha_2$. Thus, here as well,

\begin{equation}
\ca_{\oplus} = \loc(\chi_{\oplus}) \, .
\end{equation}

However, while in these examples the $\chi$-local operators already form a von Neumann algebra, this is not the case in general. 

\begin{remark}
Local operators usually only have the structure of an operator system, i.e.\, a linear subspace of $\cl(\ch)$ that is closed under the dagger and contains the identity, but that is not necessarily closed under composition. And conversely any operator system is the image of a unital completely positive map \cite{Yashin_2020} and can thus be viewed as the local operators of any Stinespring dilation of this map.
\end{remark}

%
%
%

The fact that local operators were not forming an algebra was already remarked in \cite{Arrighi2024quantumnetworks}, where an additional requirement on $\chi$ was proposed as a remedy. The following definition generalizes this idea.

\begin{definition}[$\chi$-consistency]
Let $\chi$ be a splitting map on $\ch$. We will say that an on-site operator $B \in \cl(\chlc)$ is {\em $\chi$-consistent} if $\pi^{\chi} (B \otimes \id) = (B \otimes \id) \pi^{\chi}$, or in diagrammatic notation, 

\begin{equation}
  \input{figures/cons.tikz} \, .
\end{equation}
We will call $\cons(\chi)$ the set of $\chi$-consistent operators of $\cl(\chlc)$. 
\end{definition}
This property is equivalent to asking that $\Im(\chi)$ and $\Im(\chi)^{\perp}$ are invariant for $B \otimes \id$.

\begin{proposition}\label{stable consistent}
Let $\chi$ be a splitting map on $\ch$ and $B \in \cl(\chlc)$, $B$ is $\chi$-consistent if and only if $(B \otimes \id)  (\Im(\chi)) \subseteq \Im(\chi)$ and $(B \otimes \id) (\Im(\chi)^{\perp}) \subseteq \Im(\chi)^{\perp}$.
\end{proposition}

\begin{proof}
Let $\chi$ be a splitting map on $\ch$ and $B \in \cl(\chlc)$. Suppose that $B$ is $\chi$-consistent and let $\ket{x} \in \Im(\chi)$ and $\ket{y} \in \Im(\chi)^{\perp}$. We have that $(B \otimes \id) \ket{x} = (B \otimes \id) \pi^{\chi} \ket{x} = \pi^{\chi}(B \otimes \id) \in \Im( \pi^{\chi}) = \Im(\chi)$ and that $ \pi^{\chi} (B \otimes \id) \ket{y} = (B \otimes \id) \pi^{\chi} \ket{y} = 0$, which implies that $(B \otimes \id) \ket{y} \in   \Im( \pi^{\chi})^{\perp} = \Im(\chi)^{\perp}$. Conversely, suppose that  $(B \otimes \id) (\Im(\chi)) \subseteq \Im(\chi)$ and $(B \otimes \id) (\Im(\chi)^{\perp}) \subseteq \Im(\chi)^{\perp}$. Because $\Im(\chi)$ and $\Im(\chi)^{\perp}$ span $\chlc \otimes \chrc$, it is sufficient to show that the equality holds on these subspaces which is immediate, in a similar way to the direct case.
\end{proof}

\begin{definition}[Strict locality]
Let $\chi$ be a splitting map on $\ch$ and $A \in \cl(\ch)$. We say that $A$ is {\em strictly $\chi$-local} if there exists $\tilde{A} \subset \cl(\chlc)$ such that $A \chi^{\dagger} = \chi^{\dagger}(\tilde{A} \otimes \id)$ and $\chi A = (\tilde{A} \otimes \id)\chi$, or diagrammatically,

\begin{equation}
  \input{figures/stloc.tikz} \, .
\end{equation}
We will call $\stloc(\chi)$ the set of strictly $\chi$-local operators of $\cl(\ch)$.
\end{definition}

In the case of the tensor splitting $\chi_{\otimes} : \ch = \ch_{\LL} \otimes \ch_{\R} \rightarrow \ch_{\LL} \otimes \ch_{\R}$, acting as the identity, the definition reduces to: $A$ is strictly local if and only if $A = \tilde{A} \otimes \id$, which implies that

\begin{equation}
\stloc(\chi_{\otimes}) = \loc(\chi_{\otimes}) = \ca_{\otimes} \, .
\end{equation}

Let us now look at $\stloc(\chi_{\oplus})$. Let $A$ (which we can write  $A = A_{1,1} \oplus A_{1,2} \oplus A_{2,1} \oplus A_{2,2}$ according to the decomposition of the underlying space $\ch = \ch_1 \oplus \ch_2$) be a strictly $\chi_{\oplus}$-local operator. The equality $\chi_{\oplus}A = (\tilde{A} \otimes \id)\chi_{\oplus}$ tells us that for any $\ket{x} = \ket{x_1} \oplus \ket{x_2}$, we have $(A_{1,1}\ket{x_1} + A_{2,1}\ket{x_2}) \otimes \ket{\emptyset_1}) + (\ket{\emptyset_2} \otimes (A_{1,2}\ket{x_1} + A_{2,2}\ket{x_2})) = (\tilde{A}\ket{x_1}  \otimes \ket{\emptyset_1}) + (\tilde{A} \ket{\emptyset_2} \otimes \ket{x_2})$, which implies that $A_{1,2} = 0$ and $A_{2,1} = 0$; it follows that $A \in \ca_{\oplus}$. Conversely, any element of $\ca_{\oplus}$ is a strictly $\chi_{\oplus}$-local operator: this can be witnessed by taking $\tilde{A} = \tilde{A}_1 \oplus \tilde{A}_2$ where $\tilde{A}_1$ is its action on $\ch_1$ and $\tilde{A}_2 = \alpha_2$ is its homothety ratio on $\ch_2$. Therefore we can conclude that

\begin{equation}
\stloc(\chi_{\oplus}) = \loc(\chi_{\oplus}) = \ca_{\oplus} \, .
\end{equation}

Observe that in these particular examples, the sets of local and strictly local operators were equal, and also that the two equations of the definition of strictly local operators were redundant. These two facts are not true in general but are degeneracies due to the fact that $\chi_{\otimes}$ and $\chi_{\oplus}$ belong to a special class of splitting maps that we will call balanced and study in Section \ref{balanced}. 

\begin{proposition}
Let $\chi$ be a splitting map on $\ch$ and $A \in \cl(\ch)$. $A$ is strictly $\chi$-local if and only if there exists a $\chi$-consistent $\tilde{A}$ such that $A = \chi^{\dagger} (\tilde{A} \otimes \id) \chi$, i.e.\ if it is local and has a $\chi$-consistent on-site representative.
\end{proposition}

\begin{proof}
We will prove this proposition diagrammatically. Suppose that $A$ is strictly $\chi$-local; then it is local as 

\begin{equation}
  \input{figures/stlocloc.tikz} \, ,
\end{equation}
 and moreover $\tilde{A}$ is (left) $\chi$-consistent as 
 
 \begin{equation}
  \input{figures/stloccons.tikz} \, .
\end{equation}
 
 Let us now prove the other direction. Suppose that there exists $\tilde{A}$, left $\chi$-consistent, such that $A = \chi^{\dagger} (\tilde{A} \otimes \id) \chi$. Then 
 
 \begin{equation}
  \input{figures/conslocstloc.tikz} \, ,
\end{equation}
and the other equality is proved symmetrically.
\end{proof}

\begin{remark}
Observe that we have more or less implicitly defined the notions of $\chi$-locality, $\chi$-consistency, strict $\chi$-locality on the left branch of the isometry. But symmetric definitions can be given on the right branch of the isometry and they satisfy all the properties that we have proved for the former. If we don't state otherwise, we will keep talking about the left branch but if what happens on both branches comes to play we will then talk about left $\chi$-local operators, denoted $\loc_{\LL}(\chi)$, and right $\chi$-local operators, denoted $\loc_{\R}(\chi)$ (and about other notions in a similar fashion).
\end{remark}

\subsection{Comprehension : a preorder for splitting maps}\label{comprehension}

While we have defined splitting maps as isometries that split a system into a left part and a right part, it is also possible to understand them as objects that extract the subsystem associated to the left branch from the global system. It then becomes very natural to wonder whether one can define some kind of order relation on splitting maps that would eventually match the idea that a subsystem (extracted by the left branch of $\zeta$) is included in another (extracted by the left branch of $\chi$). We thus define the following notion of inclusion between splitting maps that we will call comprehension.

\begin{definition}[Comprehension]
Let $\chi$ and $\zeta$ be two splitting maps on $\ch$. We say that $\zeta$ is {\em comprehended} in $\chi$ and write $\zeta \sqsubseteq \chi$ if there exists a Hilbert space $\ch_{\MM}$ and isometries $\tikzcircle{4pt} : \chrz \longrightarrow \ch_{\MM} \otimes \chrc$ and $\bigcirc : \chlc \longrightarrow \chlz \otimes \ch_{\MM}$ such that $(\id_{\chlz} \otimes \tikzcircle{4pt}) \zeta = (\bigcirc \otimes \id_{\chrc}) \chi$, or diagrammatically,

\begin{equation}
  \input{figures/comp.tikz} \, .
\end{equation}
\end{definition}

\begin{proposition}
Comprehension is a preorder relation on the set of splitting maps on $\ch$.
\end{proposition}

\begin{proof}
Comprehension is reflexive as for any splitting map $\chi$ on $\ch$, we can choose $\ch_{\MM} = \mathbb{C}$, $\tikzcircle{4pt} : x \longrightarrow 1 \otimes x$ and $\bigcirc : y \longrightarrow y \otimes 1$ and get that $(\id_{\chlz} \otimes \tikzcircle{4pt}) \zeta = (\bigcirc \otimes \id_{\chrc}) \chi$. To see that it is transitive, consider $\xi \sqsubseteq \zeta \sqsubseteq \chi$ three splitting maps on $\ch$; by definition there exists Hilbert spaces $\ch_{{\MM}_1}$,$\ch_{{\MM}_2}$ and dot isometries such that:

\begin{equation}
 \scalebox{0.9}{ \input{figures/transitivity1.tikz}} \, .
\end{equation}
It follows that 

\begin{equation}
  \scalebox{0.95}{ \input{figures/transitivity2.tikz}} \, ,
\end{equation}
and the conclusion is immediate by taking $\ch_{\MM} = \ch_{{\MM}_1} \otimes \ch_{{\MM}_2}$ and for isometries the compositions of dot isometries as above.
\end{proof}

\section{Equivalence between diagrammatic and algebraic subsystems}

\subsection{From diagrammatic subsystems to algebraic subsystems}

Let us now prove a few structural properties of these notions, which will allow us to link splitting maps to von Neumann algebras.

\begin{definition}
Give a splitting map $\chi : \ch \rightarrow \chlc \otimes \chrc$, we define $\sigma^{\chi} : \cl(\chlc) \longrightarrow \cl(\ch)$ as $\sigma^{\chi}(B) = \chi^{\dagger} (B \otimes \id_{\chrc}) \chi$. It is the map that takes on-site representatives to the corresponding $\chi$-local operators acting on the whole Hilbert space. Diagrammatically, we can see it as 

\begin{equation}
  \input{figures/sigma.tikz} \, .
\end{equation}
It is a linear operator such that $\sigma^{\chi} (A^{\dagger}) = \sigma^{\chi}(A)^{\dagger}$.
Moreover it satisfies 
\begin{itemize}
\item $\loc(\chi) = \sigma^{\chi}(\cl(\chlc))$
\item $\stloc(\chi) = \sigma^{\chi}(\cons(\chi))$
\end{itemize}
\end{definition}

\begin{proposition}
The set of $\chi$-consistent operators, $\cons(\chi)$, is a von Neumann algebra included in $\cl(\chlc)$. It is the biggest von Neumann algebra $\ca$ in $\cl(\chlc)$ such that $\sigma^{\chi}_{|\ca}$ is a $*$-algebra homomorphism.
\end{proposition}

\begin{proof}
First, let us Observe that $\cons(\chi)$ is a linear subspace of $\cl(\chlc)$ containing $\id_{\chlc}$; furthermore it is stable by composition since for all $A,B$ in $\cons(\chi)$, $\pi^{\chi}(AB \otimes \id) = (A \otimes \id )\pi^{\chi} (B \otimes \id)  = (AB \otimes \id) \pi^{\chi}$; finally it is stable by adjoint as for all $A$ in $\cons(\chi)$, $\pi^{\chi}(A^{\dagger} \otimes \id) =  ((A \otimes \id) \pi^{\chi})^{\dagger} = (\pi^{\chi}(A \otimes \id))^{\dagger} = (A^{\dagger} \otimes \id) \pi^{\chi}$. This proves that $\cons(\chi)$ is a von Neumann subalgebra of $\cl(\chlc)$. Let us now prove that $\sigma^{\chi}_{|\ca}$ is a $*$-algebra homomorphism. It is linear and preserves the adjoint as $\chi^{\dagger} (A^{\dagger} \otimes \id) \chi = (\chi^{\dagger} (A \otimes \id) \chi)^{\dagger}$. It remains to show that it preserves the product, which we can show diagrammatically as

\begin{equation}
  \input{figures/consmorphism.tikz} \, .
\end{equation}

Finally let us prove that it is the biggest von Neumann algebra $\ca$ of $\cl(\chlc)$ such that $\sigma^{\chi}_{|\ca}$ is  a $*$-algebra homomorphism. Let $\ca \subseteq \cl(\chlc)$ be a von Neumann algebra satisfying this property. Let $A \in \ca$ and $\overline{\pic} = \id - \pi^{\chi}$. Then $\pic (A \otimes \id)^{\dagger} (A \otimes \id)  \pic = \pic (A \otimes \id)^{\dagger} \pic (A \otimes \id) \pic$, thus $\parallel (A \otimes \id) \pic \parallel^2 = \Tr(\pic (A \otimes \id)^{\dagger} (A \otimes \id) \pic) = \Tr(\pic (A \otimes \id)^{\dagger} \pic (A \otimes \id) \pic) = \parallel \pic (A \otimes \id) \pic \parallel^2$ and because $(A \otimes \id) \pic = \pic (A \otimes \id)\pic + \overline{\pic}  (A \otimes \id) \pic$ with $\langle \pic (A \otimes \id) \pic ,  \overline{\pic} (A \otimes \id) \pic \rangle = 0$, by the Pythagorean theorem $\parallel  (A \otimes \id) \pic \parallel^2 = \parallel \pic (A \otimes \id) \pic \parallel^2 + \parallel \overline{\pic} (A \otimes \id) \pic \parallel^2$ which implies that $ \parallel \overline{\pic}  (A \otimes \id) \pic \parallel^2 = 0$ and thus that $ \overline{\pic} (A \otimes \id) \pic = 0$. This tells us that $(A \otimes \id) \pic = \pic (A \otimes \id) \pic $. The same reasoning can be done with $\Tr(\pic (A \otimes \id) (A \otimes \id)^{\dagger} \pic) = \Tr(\pic (A \otimes \id) \pic (A \otimes \id)^{\dagger} \pic)$; it follows that $\pic (A \otimes \id) = \pic (A \otimes \id) \pic $, which concludes the proof.
\end{proof}

\begin{corollary}
The set of strictly $\chi$-local operators, $\stloc(\chi)$, is a von Neumann subalgebra of $\cl(\ch)$.
\end{corollary}

This result tells us that, starting from a splitting map, we can recover an algebraic subsystem as its set of strictly local operators. Let us now show that comprehension between splitting maps implies the respective inclusions of their local and strictly-local operators.

\begin{proposition}\label{comprehension implies inclusion}
Let $\chi$ and $\zeta$ be two splitting maps on $\ch$. If $\zeta \sqsubseteq \chi$, then:
\begin{enumerate}
\item $\loc(\zeta) \subseteq \loc(\chi)$
\item $\stloc(\zeta) \subseteq \stloc(\chi)$
\end{enumerate}
\end{proposition}

\begin{proof}
Let $A \in \loc(\zeta)$. Writing

\begin{equation}
  \input{figures/comploc.tikz}
\end{equation}
proves that  $A \in \loc(\chi)$ and thus that $\loc(\zeta) \subseteq \loc(\chi)$. Suppose now that $\tilde{A} \in \cons(\zeta)$ and thus that $A \in \stloc(\zeta)$, we get that 

\begin{equation}
  \scalebox{0.9}{  \input{figures/compcons.tikz}} \, ,
\end{equation}
which proves that $\bigcirc^{\dagger} (\tilde{A} \otimes \id) \bigcirc \in \cons(\chi)$ and thus that $A \in \stloc(\chi)$.
\end{proof}

As a consequence of these results, we have that the assignment of the set of strictly-local operators to any splitting map gives a preorder preserving map from splitting maps to von Neumann algebras,
\[
\stloc(\cdot): (\spl(\ch),\sqsubseteq) \longrightarrow (\vnalg(\ch), \subseteq). 
\]
This in-turn sets up one-half of a possible equivalence between von Neumann algebras and splitting maps.
\subsection{From algebraic subsystems to diagrammatic subsystems}

In the previous sections, we defined splitting maps as a tool to select a subsystem and showed that we can associate to each a von Neumann subalgebra of $\cl(\ch)$, through the preorder preserving map $\stloc(\cdot): (\spl(\ch),\sqsubseteq) \longrightarrow (\vnalg(\ch),\subseteq)$. In this section we will first look at the reverse direction and show that we can associate to any von Neumann subalgebra of $\cl(\ch)$ a canonical splitting map, provided by the Artin-Wedderburn theorem, and for which locality and strict locality coincide.

\begin{proposition}
Let $\ca$ be a unital subalgebra of $\cl(\ch)$ and $U : \ch \longrightarrow \bigoplus_i (\ch_{{\LL}}^{i} \otimes \ch_{{\R}}^{i})$ be a representation unitary for $\ca$. Then

\begin{align}
\chi_U :  & \ch \longrightarrow \bigoplus_i (\ch_{{\LL}}^{i} \otimes \ch_{{\R}}^{i}) \hookrightarrow (\oplus_i \ch_{{\LL}}^{i}) \otimes (\oplus_i \ch_{{\R}}^{i}) \\
& \ket{x} \longrightarrow U(\ket{x})
\end{align}
is a splitting map such that $\stloc_{\LL}(\chi_U) = \loc_{\LL}(\chi_U) = \ca$ and $\stloc_{\R}(\chi_U) = \loc_{\R}(\chi_U) = \ca'$. We will call $\chi$ a canonical splitting map for $\ca$.
\end{proposition}

\begin{proof}
First we Observe that $\Im(\chi_U) =  \bigoplus_i (\ch_{{\LL}}^{i} \otimes \ch_{{\R}}^{i}) \subseteq (\oplus_i \ch_{{\LL}}^{i}) \otimes (\oplus_i \ch_{{\R}}^{i})$ and thus that $\pi^{\chi_U} = \bigoplus_i \id_{\ch_{{\LL}}^{i}} \otimes \id_{\ch_{{\R}}^{i}} $, as it is the projector on $\Im(\chi_U)$. Then the left $\chi_U$-consistent operators are the operators $A$ of $\cl(\ch_{\chi_U}^{{\LL}}) = \cl( \bigoplus_i \ch_{{\LL}}^{i})$ such that $A \otimes \id_{ \bigoplus_i \ch_{{\R}}^{i}}$ commutes with $\pi^{\chi_U} = \bigoplus_i \id_{\ch_{{\LL}}^{i}} \otimes \id_{\ch_{{\R}}^{i}}$. This means that $\bigoplus_i (A\id_{\ch_{{\LL}}^{i}}) \otimes \id_{\ch_{{\R}}^{i}} = \bigoplus_i (\id_{\ch_{{\LL}}^{i}}A) \otimes \id_{\ch_{{\R}}^{i}}$ and thus that for all $i$, $A\id_{\ch_{{\LL}}^{i}} = \id_{\ch_{{\LL}}^{i}}A = \id_{\ch_{{\LL}}^{i}}A \id_{\ch_{{\LL}}^{i}}$. It follows, because the $\id_{\ch_{{\LL}}^{i}}$ are a family of pairwise orthogonal projectors that sum to $\id_{ \bigoplus_i \ch_{{\LL}}^{i}}$, that $A = (\id_{ \bigoplus_i \ch_{{\LL}}^{i}}) A (\id_{ \bigoplus_i \ch_{{\LL}}^{i}}) = \bigoplus_i \id_{\ch_{{\LL}}^{i}} A \id_{\ch_{{\LL}}^{i}} = \bigoplus_i A_i$ with $A_i \in \cl(\ch_{{\LL}}^{i})$. Conversely, one can check that all operators of this form are left $\chi_U$-consistent.

We then compute that $\sigma^{\chi_U}(A) = \chi_U^{\dagger} (A \otimes \id)  \chi_U =  \chi_U^{\dagger} \pi^{\chi_U} (A \otimes \id) \pi^{\chi_U} \chi_U = \chi_U^{\dagger} (\bigoplus_i \id_{\ch_{\LL}^{i}}A \id_{\ch_{\LL}^{i}}\otimes \id_{\ch_{\R}^{i}}) \chi_U$, which we can write as $U^{\dagger}  (\bigoplus_i A_{i} \otimes \id_{\ch_{{\R}}^{i}})  U$. We therefore get that $\loc(\chi_U) = \sigma^{\chi_U}(\cl(\chlc)) = \sigma^{\chi_U}(\cons_{\LL}(\chi)) = \stloc(\chi_U) \subseteq \ca$. Conversely, any element of $\ca$ can be decomposed like this because all maps of the form $\bigoplus_i A_i$ with $A_i \in \cl(\ch_{{\LL}}^{i})$ are left $\chi_U$-consistent. This proves that $\stloc_{\LL}(\chi_U) = \loc_{\LL}(\chi_U)= \ca$. A symmetric reasoning proves that $\stloc_{\R}(\chi_U) = \ca'$.
\end{proof}

\begin{remark}
A given subalgebra $\ca$ of $\cl(\ch)$ doesn't have a unique canonical splitting map, but all its canonical splitting maps are equal up to some local unitaries on the left and on the right.
\end{remark}

Let us now show that a map $F : (\vnalg(\ch),\subseteq) \rightarrow (\spl(\ch),\sqsubseteq)$ that sends every $\ca \in \vnalg(\ch)$ to a canonical splitting map $\chi_{\ca}$ will be preorder preserving.

\begin{proposition}\label{inclusion implies comprehension}
Let $\ca_S \subseteq \ca_B$ be two von Neumann subalgebras of $\cl(\ch)$. Let $\chi : \ch \rightarrow \ch_{\LL}^B \otimes \ch_{\R}^B$ and $\zeta : \ch \rightarrow \ch_{\LL}^S \otimes \ch_{\R}^S$ be canonical splitting maps for $\ca_B$ and $\ca_S$ respectively. Then $\zeta \sqsubseteq \chi$.
\end{proposition}

\begin{proof}
See Appendix \ref{bigproof}.
\end{proof}

As a result of the above we have established the following: that there exist a preorder homomorphism $F$ from von Neumann algebras to splitting maps and a preorder homomorphism $G$ from splitting maps to von Neumann algebras.

\[
\begin{tikzcd}
 (\vnalg(\ch),\subseteq) \arrow[r, "F", bend left] &(\spl(\ch),\sqsubseteq)  \arrow[l, "G = \stloc(.)", bend left] 
\end{tikzcd}
\]

One might wonder whether this pair of preorder homomorphisms establishes an equivalence, such that $FG = id$. In fact, it is not even true that for each $\chi$, $FG(\chi) \cong \chi$. Indeed, a simple counter example is the following. Consider $\chi : \mathbb{C}^2 \rightarrow \mathbb{C}^4 \otimes \mathbb{C}^4$ such that $\chi \ket{0} = \frac{1}{\sqrt{2}}(\ket{00}+\ket{11})$ and $\chi \ket{1} = \frac{1}{\sqrt{2}}(\ket{20}+\ket{33})$. Its algebra of strictly (left)-local operators is $\stloc_{\LL}(\chi) = \{ \alpha \ketbra{0}{0} + \beta \ketbra{1}{1} \} \subseteq \cl(\mathbb{C}^2)$ and thus $\zeta = FG(\chi)$ will be a canonical splitting map such that $\zeta\ket{0} = \ket{00}$ and $\zeta\ket{1} = \ket{11}$. Now suppose that $\chi \sqsubseteq \zeta$ and thus that $(\id \otimes \tikzcircle{4pt}) \chi  = (\bigcirc \otimes \id) \zeta$. Applying this equality to $\ket{0}$, we find that $\tikzcircle{4pt}\ket{0} = \ket{m_0} \otimes \ket{0}$, and applying it to $\ket{1}$ yields $\tikzcircle{4pt}\ket{0} = \ket{m_1} \otimes \ket{1}$. This is not possible and thus contradicts the fact that $\chi \sqsubseteq \zeta = FG(\chi)$. 

Nonetheless, it is natural to wonder which splitting maps satisfy the property that $FG(\chi) \cong \chi$, since keeping only those splitting maps would establish the widest possible equivalence (based on strictly local operators) between splitting maps and von Neumann algebras. 

\begin{remark}\label{equivalence implies balanced}
If $\chi$ is a splitting map such that $FG(\chi) \cong \chi$, i.e. $\chi \sqsubseteq F \circ G(\chi) \sqsubseteq \chi$, then by definition of $G$, $\zeta = F \circ G(\chi)$ is a canonical splitting map for the algebra $\ca = \stloc_{\LL}(\chi)$ and thus such that $\stloc_{\LL}(\zeta) = \stloc_{\R}(\zeta)'$. Then because $\chi \sqsubseteq \zeta$, $\stloc_{\LL}(\chi) \subseteq \stloc_{\LL}(\zeta)$ and $\stloc_{\R}(\zeta) \subseteq \stloc_{\R}(\chi)$. And because $\zeta \sqsubseteq \chi$, $\stloc_{\LL}(\zeta) \subseteq \stloc_{\LL}(\chi)$ and $\stloc_{\R}(\chi) \subseteq \stloc_{\R}(\zeta)$. It follows that $\stloc_{\LL}(\chi) = \stloc_{\LL}(\zeta)$ and that $\stloc_{\R}(\zeta) =  \stloc_{\R}(\chi)$ and then that
\[
\stloc_{\LL}(\chi) = \stloc_{\R}(\chi)' .
\]
\end{remark}

This leads naturally to the question: can those $\chi$ that are balanced in the above sense be viewed as equivalent up to preorder to von Neumann algebras? 

\subsection{Balanced splitting maps}\label{balanced}

We will focus in this section on the class of these balanced splitting maps. First, let us Observe that this is not a trivial requirement: in general, the strictly left local operators and strictly right local operators of a splitting map only have to be commuting subalgebras, and not precisely each other's commutant. Indeed, let us consider the splitting map $\chi : \mathbb{C}^{2} \rightarrow \mathbb{C}^{2} \otimes \mathbb{C}^{2}$ defined by $\chi (\ket{0}) = \ket{00}$ and $\chi (\ket{1}) = \ket{10}$: it leads to $\stloc_{\LL}(\chi) = \{\alpha \ketbra{0}{0} + \beta \ketbra{1}{1}\}$ and $\stloc_{\R}(\chi) = \mathbb{C} \id$, which are not each other's commutants as $\stloc_{\LL}(\chi)' = \stloc_{\LL}(\chi) \neq \stloc_{\R}(\chi)$.

We will start by defining balanced splitting maps, then show that they can be decomposed in a way that is related to canonical splitting maps. Finally we will prove that they yield the wanted equivalence with von Neumann algebras.

\begin{definition}[Balanced splitting maps]
A splitting map $\chi$ is said to be {\em balanced} if $\stloc_{\R}(\chi) = \stloc_{\LL}(\chi)'$. We will call $\balanced(\ch)$ the set of balanced splitting maps on $\ch$.
\end{definition}

We also define a subclass of balanced splitting maps that we will call lean and that are encoding the algebras of strictly local operators in a minimal way, i.e. such that the corresponding on-site algebras of consistent operators have no redundancies.

\begin{definition}[Lean splitting maps]
A splitting map is said to be {\em lean} if it is balanced and its algebras of $\chi$-consistent operators are such that $\cons_{\LL}(\chi)' = \cz(\cons_{\LL}(\chi)) \subseteq \cl(\chlc)$ and $\cons_{\R}(\chi)' = \cz(\cons_{\R}(\chi)) \subseteq \cl(\chrc)$. We will call $\lean(\ch)$ the set of lean splitting maps on $\ch$.
\end{definition}

We will now give a general result on the decomposition of splitting maps. It will allow us to reconstruct the shape of balanced splitting maps from the structure of the ones whose algebras of strictly local operators are factors. 

\begin{lemma}\label{decomposition of split}
Let $\chi \in \spl(\ch)$ and let $\Atproj (\cz(\stloc_{\LL}(\chi))) = \{ \pi_i\}_{i \in I}$. Then there exists a family of Hilbert spaces $\{ \ch_i = \pi_i \ch \}_{i \in I}$ and isometries $\chi_i = \chi \pi_i : \ch \longrightarrow \chlc \otimes \chrc$ such that: 
\begin{itemize}
\item $\ch = \bigoplus_i \ch_i$,
\item $\chi = \sum_i \chi_i =  \sum_i (\chi_i)_{|\ch_i}$,
\item for all $i$ in $I$, $(\chi_i)_{|\ch_i}$ is a splitting map on $\ch_i$,
\item its algebra of strictly local operators (seen as a subalgebra of $\cl(\ch)$) is $\pi_i \stloc_{\LL}(\chi)$ and is thus a factor,
\item there exists a family $\{\ch_{\LL}^{i}\}$ of pairwise orthogonal subspaces of $\chlc$ such that $\Im(\chi_i) \subseteq \ch_{\LL}^{i} \otimes \chrc$.
\end{itemize}
\end{lemma}

\begin{proof}
The first two points are immediate because the atomic projectors of a commutative algebra sum to the identity. Then, because the $\pi_i$ are orthogonal projectors, they are isometries when restricted to their image and thus $(\chi_i)_{|\ch_i} = (\chi \pi_i)_{|\ch_i} = \chi (\pi_i)_{|\ch_i}$ is a splitting map.

Let us now prove that its algebra of strictly local operators is $\pi_i \stloc_{\LL}(\chi)$. Let $A \in \cons_{\LL}(\chi)$, we get that $\sigma^{\chi}(A) = \chi^{\dagger} (A \otimes \id) \chi \in \stloc_{\LL}(\chi)$ commutes with $\pi_i$ by definition of the atomic projectors. It follows that $\chi \chi^{\dagger} (A \otimes \id) \chi \pi_i \chi^{\dagger} = \chi \pi_i  \chi^{\dagger} (A \otimes \id) \chi \chi^{\dagger}$ and then that $ (A \otimes \id) \chi \pi_i \chi^{\dagger} = \chi \pi_i  \chi^{\dagger} (A \otimes \id)$ because $A$ is left $\chi$-consistent. Moreover, Observe that $\chi_i \chi_i^{\dagger} = \chi \pi_i \chi^{\dagger}$ and thus $A$ is left $\chi_i$-consistent. It follows that $\pi_i \sigma^{\chi}(A) = \sigma^{\chi_i}(A) \in \stloc_{\LL}(\chi_i)$, which proves that $\pi_i \stloc_{\LL}(\chi) \subseteq \stloc_{\LL}(\chi_i)$.

It remains to show the reverse inclusion. By Proposition \ref{prop: kernels are blocks}, we know that there exists an orthogonal projector $\mu$ (sum of atomic projectors of $\cz(\cons_{\LL}(\chi))$) such that $\sigma^{\chi}_{|\mu \cons_{\LL}(\chi)} : \mu \cons_{\LL}(\chi) \longrightarrow \stloc_{\LL}(\chi)$ is an isomorphism of von Neumann algebras. In particular it is such that the $\tilde{\pi}_i  = (\sigma^{\chi}_{|\mu \cons_{\LL}(\chi)})^{-1}\pi_i$ are atomic projectors of $\cz(\cons_{\LL}(\chi))$ (they're precisely the ones that are not killed by $\sigma^{\chi}$). Let $A \in \cons_{\LL}(\chi_i)$, i.e. such that $ (A \otimes \id) \chi \pi_i \chi^{\dagger} = \chi \pi_i \chi^{\dagger} (A \otimes \id)$. By definition of the $\tilde{\pi}_i$, this implies that $(A \otimes \id) \chi \chi^{\dagger} (\tilde{\pi}_i \otimes \id) \chi \chi^{\dagger} =  \chi \chi^{\dagger} (\tilde{\pi}_i \otimes \id) \chi \chi^{\dagger} (A \otimes \id)$ and then that $(\tilde{\pi}_i A \otimes \id) \chi \chi^{\dagger} =  \chi \chi^{\dagger} (\tilde{\pi}_i A \otimes \id)$ meaning that $\tilde{\pi}_i A$ is left $\chi$-consistent. It follows that $ \sigma^{\chi_i}(A) = \chi_i^{\dagger} (A \otimes \id) \chi_i = \pi_i \chi^{\dagger} (A \otimes \id) \chi \pi_i = \pi_i \chi^{\dagger}  (\tilde{\pi}_i A \otimes \id) \chi = \pi_i \sigma^{\chi}(\tilde{\pi}_i A) \in \pi_i \stloc_{\LL}(\chi)$ proving the fourth point.

Let us finally prove the last point. By definition the $\tilde{\pi_i}$ are $\chi$-consistent on-site representatives of the $\pi_i \in \stloc_{\LL}(\chi)$, which yields that for all $i$, $\chi_i = \chi \pi_i = (\tilde{\pi_i} \otimes \id) \chi$. Then, as we have seen, the $\tilde{\pi_i}$ are atomic projectors of the algebra $\cons_{\LL}(\chi)$ and are thus pairwise orthogonal projectors. Taking $\ch_{\LL}^{i} = \Im(\tilde{\pi_i})$ concludes the proof. 

\end{proof}

Observe that a similar \lq \lq decomposition\rq \rq can be done with respect to $\stloc_{\R}(\chi)$ and that it will in general be different from the first one, because the two algebras have different centers which thus have a priori different atomic projectors. However, note that when $\chi$ is balanced, i.e.\ when $\stloc_{\LL}(\chi) = \stloc_{\R}(\chi)'$, the two algebras share the same center  and thus the two decompositions are the same. This allows us to infer the shape of general balanced splitting maps from the ones whose algebras of strictly local operators are factor. Let us thus focus on the latter ones.

\begin{lemma}\label{shape of good split factor}
Let $\chi \in \balanced(\ch)$ and suppose that $\ca = \stloc_{\LL}(\chi)$ is a factor. Then there exists an orthonormal basis $\{ \ket{lr} \}$ of $\ch$ and families of orthonormal elements $\{ \ket{l+m} \}$ in $\chlc$ and $\{ \ket{r+m} \}$ in $\chrc$ such that $\ket{\phi_{lr}} = \chi(\ket{lr}) = \sum_m \lambda_m \ket{l+m}\ket{r+m}$.
\end{lemma}

\begin{proof}
First remember that by Theorem \ref{AW} there exists a unitary isomorphism $\ch \cong \bigoplus_i \ch_{\LL}^{i} \otimes  \ch_{\R}^{i}$ such that $\ca \cong \bigoplus_i \cl(\ch_{\LL}^{i}) \otimes  \id_{\ch_{\R}^{i}}$ (and symmetrically for $\ca'$). Because $\ca$ is a factor, this sum has a unique term $\ch \cong \ch_{\LL} \otimes \ch_{\R}$ and $\ca \cong \cl(\ch_{\LL}) \otimes  \id_{\ch_{\R}}$. Let's make a choice of orthonormal basis $\{ \ket{l }\}$ of $\ch_{\LL}$ and $\{ \ket{r }\}$ of $\ch_{\R}$; for the rest of the proof we will identify $\ch$ with $\ch_{\LL} \otimes \ch_{\R}$ equipped with the orthonormal basis $\{ \ket{l} \otimes \ket{r} \}$. Then by Proposition \ref{prop: kernels are blocks} there exists an orthogonal projector $\mu \in \cons_{\LL}(\chi)$ such that $\sigc_{|\mu(\cons_{\LL}(\chi))} : \mu(\cons_{\LL}(\chi)) \longrightarrow \stloc(\chi)$ is an isomorphism of von Neumann algebras.

For all $A\in \ca$ we define $\tilde{A} = (\sigc_{|\mu(\cons_{\LL}(\chi))})^{-1} (A) \in  \mu(\cons_{\LL}(\chi))$. Consider $\ket{\phi_{00}} = \chi(\ket{00})$; by Hilbert-Schmidt decomposition it can be decomposed as $\ket{\phi_{00}}= \sum_i \lambda^{00}_i \ket{i_{\LL}} \ket{i_{\R}}$ where $\{ \ket{i_{\LL}} \}_{i  \in I}$ is an orthonormal family of $\chlc$ and $\{ \ket{i_{\R}} \}_{i  \in I}$ is one of $\chrc$. Similarly $\ket{\phi_{10}} = \chi(\ket{10})$ can be decomposed as $\ket{\phi_{10}} = \sum_j \lambda^{10}_j \ket{j_{\LL}} \ket{j_{\R}}$. Let us now define $A = \ketbra{0}{0} \otimes \id_{\ch_{\R}}$ and $B = \ketbra{1}{1} \otimes \id_{\ch_{\R}}$, which are  two orthogonal and pairwise orthogonal projectors of $\cl(\ch)$. Because $\sigc_{|\mu(\cons_{\LL}(\chi))}$ is an isomorphism of  von Neumann algebras, $\tilde{A}$ and $\tilde{B}$ are also orthogonal and pairwise orthogonal projectors of $\cons_{\LL}(\chi)$ such that $(\tilde{A}\otimes \id) \ket{\phi_{00}} = \ket{\phi_{00}}$ and $ (\tilde{B} \otimes \id)\ket{\phi_{10}} = \ket{\phi_{10}}$, i.e. $\ket{\phi_{00}} = \sum_i \lambda^{00}_i (\tilde{A} \ket{i_{\LL}}) \ket{i_{\R}}$ and $\ket{\phi_{10}} = \sum_j \lambda^{10}_i (\tilde{B} \ket{j_{\LL}}) \ket{j_{\R}}$. Then because the $\{ \ket{i_{\R}} \}$ and $\{ \ket{j_{\R}} \}$ are families of orthonormal vectors, we obtain that for all $i$, $\ket{i_{\LL}} \in \Im(\tilde{A})$, and that for all $j$, $\ket{j_{\LL}} \in \Im(\tilde{B})$. As the two projectors are orthogonal we deduce that for all $i,j$, $\braket{i_{\LL}}{j_{\LL}} = 0$.

Now consider $C = \ketbra{1}{0} \otimes \id \in \ca$; its counterpart in $\mu(\cons_{\LL}(\chi))$, $\tilde{C}$, satisfies the same properties: $\tilde{C}^{\dagger} \tilde{C} = \tilde{A}$,  $\tilde{C}\tilde{C}^{\dagger} = \tilde{B}$ and $ \tilde{C}\ket{\phi_{00}} =\ket{\phi_{10}}$. The two first properties imply that $\tilde{C}$ acts unitarily between $\Im(\tilde{A})$ and $\Im(\tilde{B})$. The third one tells us that $\ket{\phi_{10}} = \sum_i \lambda_i^{00} (\tilde{C}\ket{i_{\LL}}) \ket{i_{\R}}$ and then by unitarity that the $ \tilde{C}\ket{i_{\LL}}$ form an orthonormal family. It follows that we can rewrite it $\ket{\phi_{10}} = \sum_i \lambda_i^{00} \ket{(i+1)_{\LL}} \ket{i_{\R}}$. This reasoning generalizes to any choice of basis element on the left and on the right symmetrically, allowing us to write for all $l$ and $r$, $\ket{\phi_{lr}} = \sum_m \lambda_m \ket{l+m}\ket{r+m}$.
\end{proof}

This proof can be understood as using the fact that $\chi : \ch_{\LL} \otimes \ch_{\R} \longrightarrow \chlc \otimes \chrc$ is a non-signalling perfect channel. The result of this lemma may seem a bit abstract but it can be stated more concretely when relating the shape of such splitting maps to canonical ones, as done in the next corollary.

\begin{corollary}
Let $\chi \in \balanced(\ch)$ and suppose that $\ca = \stloc_{\LL}(\chi)$ is factor, then there exists $\zeta : \ch \longrightarrow \ch_{\LL} \otimes \ch_{\R}$ a canonical splitting map for $\ca$, a shared (entangled) state $\phi \in \ch_{\MM} = \ch_{{\MM}_{\LL}} \otimes \ch_{{\MM}_{\R}}$ and isometries $U_{\LL} : \ch_{\LL} \otimes \ch_{{\MM}_{\LL}} \longrightarrow \chlc$ and $U_{\R} :  \ch_{{\MM}_{\R}} \otimes \ch_{\R} \longrightarrow \chrc$ such that 

\begin{equation}
  \input{figures/factorChiDecomposition.tikz}
\end{equation}

\end{corollary}

\begin{proof}
Using Lemma \ref{shape of good split factor}, the result is immediate by taking $\zeta$ the canonical splitting map associated to the unitary $U : \ket{lr} \longrightarrow \ket{l} \otimes \ket{r} \in \ch_{\LL} \otimes \ch_{\R}$ (that appears in the proof), $\phi = \sum_m \lambda_m \ket{m}\ket{m}$, $U_{\LL} : \ket{l} \otimes \ket{m} \longrightarrow \ket{l+m}$ and $U_{\R} : \ket{m} \otimes \ket{r} \longrightarrow \ket{r+m}$.
\end{proof}

Finally let us come back to general balanced splitting maps.

\begin{corollary}\label{shape of good split}
Let $\chi \in \balanced(\ch)$ and $\ca = \stloc(\chi)$. Let $\{\pi_i\}$ be the set of atomic projectors of $\cz(\ca)$, then there exists a family of canonical splitting maps for $\pi_i \ca$, $\zeta_i : \pi_i \ch \longrightarrow \ch_{{\LL}_i}\otimes \ch_{{\R}_i}$, a family of entangled states $\phi_i \in \ch_{{\MM}_i} = \ch_{{\MM}_{{\LL}_i}} \otimes \ch_{{\MM}_{{\R}_i}}$ and families of isometries $U_{{\LL}_i} : \ch_{{\LL}_i} \otimes \ch_{{\MM}_ {{\LL}_i}} \longrightarrow \chlc$ and $U_{{\R}_i} :  \ch_{{\MM}_{{\R}_i}} \otimes \ch_{{\R}_i} \longrightarrow \chrc$ such that 

\begin{equation}
  \input{figures/chiDecomposition.tikz}
\end{equation}

\end{corollary}

\begin{proof}
By Lemma \ref{decomposition of split} there exists a family of splitting maps $\chi_i$ such that $\chi = \sum_i (\chi_i)_{|\ch_i}$. Because $\stloc_{\LL}(\chi) = \stloc_{\R}(\chi)'$, both algebras have the same centre and thus the decomposition of $\chi$ for the left algebra is also the decomposition of $\chi$ for the right algebra which implies that $\stloc_{\LL}((\chi_i)_{|\ch_i}) = \pi_i \stloc_{\LL}(\chi) = \pi_i \stloc_{\R}(\chi)' =  \pi_i (\pi_i \stloc_{\R}(\chi))' = \pi_i \stloc_{\R}((\chi_i)_{|\ch_i})' = \stloc_{\R}((\chi_i)_{|\ch_i})' $ with $\stloc_{\LL}((\chi_i)_{|\ch_i})$ and $\stloc_{\R}((\chi_i)_{|\ch_i})$ factor. Then Lemma \ref{shape of good split factor} gives us, for all $i$, the existence of the $\zeta_i $, $\phi_i$, $U_{{\LL}_i}$, $U_{{\R}_i}$ and summing them concludes the proof.
\end{proof}

Observe that the balanced splitting maps have this property that we had witnessed with the examples $\chi_{\otimes}$ and $\chi_{\oplus}$ in section \ref{section3}: there is no difference between their local and strictly local operators. Indeed, let $\chi \in \balanced(\ch)$, we have that $\stloc_{\LL}(\chi) \subseteq \loc_{\LL}(\chi) \subseteq \stloc_{\R}(\chi)' = \stloc_{\LL}(\chi)$ with the inclusions being true in general and the equality being the definition of $\chi$ being balanced. This proves that $\loc_{\LL}(\chi) = \stloc_{\LL}(\chi)$ and a similar equality for the right local operators can be proved symmetrically.

Before proving the equivalence between balanced splitting maps and von Neumann algebras, let us focus briefly on lean splitting maps and show that the algebraic requirement of being lean makes the shared states $\phi_i$ in Corollary \ref{shape of good split} separable; meaning that lean splitting maps are none other than canonical splitting maps followed by some local isometries.

\begin{proposition}\label{decomposition of lean}
Let $\chi \in lean(\ch)$ and $\ca = \stloc(\chi)$. Then, there exists a canonical splitting map $\zeta$ for $\ca$ and isometries $U_{\LL}$ and $U_{\R}$ such that $\chi = (U_{\LL} \otimes U_{\R})\zeta$, i.e.

\begin{equation}
  \input{figures/chiDecomposition2.tikz}
\end{equation}

\end{proposition}

\begin{proof}
Let $\chi \in \lean(\ch)$ and $\ca = \stloc(\chi)$. We know, by Lemma \ref{decomposition of split}, that we can decompose $\ch = \bigoplus_i \ch_i = \bigoplus \pi_i \ch$ and $\chi = \sum_i \chi_i = \sum_i \chi \pi_i$; and that there exist families $\{\ch_{\LL}^{i}\}$ of pairwise orthogonal subspaces of $\chlc$, and $\{\ch_{\R}^{i}\}$ of pairwise orthogonal subspaces of $\chrc$ such that $\Im(\chi_i) \subseteq \ch_{\LL}^{i} \otimes \ch_{\R}^{i}$. We also know that, for all $i$, $\ker(\chi_i) = (\ch_i)^{\perp}$ and that, by Lemma \ref{shape of good split factor}, $(\chi_i)_{|\ch_i} : \ket{lr} = \sum_m \lambda_m \ket{l+m}\ket{r+m} \in \ch_{\LL}^{i} \otimes \ch_{\R}^{i}$. Observe that, because $\{ \ch_i\}$, $\{ \ch_{\LL}^{i} \}_i$ and $\{ \ch_{\R}^{i} \}_i$ are all families of pairwise orthogonal subspaces, commuting with $\chi \chi^{\dagger}$ is precisely equivalent to commuting with each of the $\chi_i \chi_i^{\dagger}$. It follows that $\cons_{\LL}(\chi) = \bigcap_i \cons_{\LL}((\chi_i)_{|\ch_i})$. Let $A \in \cons_{\LL}((\chi_i)_{|\ch_i})$, which we know by Lemma \ref{stable consistent} to be the algebra of operators $A \in \chlc$ such that $\Im(\chi_i)$ and $\Im(\chi_i)^{\perp}$ are invariant subspaces for $A \otimes \id$. One can check that $A$ is of the form $U_{{\LL}_i} (\tilde{A} \otimes \id) U_{{\LL}_i}^{\dagger} \oplus \tilde{B}$ where $U_{{\LL}_i}$ is the isometry (defined in Corollary \ref{shape of good split}) such that $U_{{\LL}_i} : \ket{l}\ket{m} \rightarrow \ket{l+m}$ and $\tilde{B}$ an operator in  $\cl( \Span(\{ \ket{l+m} \}) ^{\perp})$. Conversely, one can also check that any $A$ of this form is $\chi_i$-consistent.

Let us now suppose that there exists $i$ such that $\phi_i$ (also defined in Corollary \ref{shape of good split}) is not separable, or equivalently such that the corresponding decomposition $(\chi_i)_{|\ch_i} : \ket{lr} = \sum_m \lambda_m \ket{l+m}\ket{r+m}$ has more than one term. Let us call $m_0$ and $m_1$ the two first indices and consider the non-trivial map $C = U_{{\LL}_i} (\id \otimes \ketbra{m_0}{m_1}) U_{{\LL}_i}^{\dagger}$. This operator commutes by construction with every $\chi_i$-consistent operator. It follows that $C \in \cons_{\LL}((\chi_i)_{|\ch_i})' \subseteq \cons_{\LL}(\chi)'$. However $C \notin \cons_{\LL}((\chi_i)_{|\ch_i})$  and thus $C \notin \cons_{\LL}(\chi)$, which contradicts the assumption of $\chi$ being lean. It follows that for all $i$, $\phi_i$ is separable or equivalently that $(\chi_i)_{|\ch_i} : \ket{lr} \rightarrow \ket{l+0} \ket{r+0}$. Defining $\zeta : \ket{lr} \rightarrow \ket{l} \otimes \ket{r}$, $U_{\LL} : \ket{l} \rightarrow \ket{l+0}$ and $U_{\R} : \ket{r} \rightarrow \ket{r+0}$ provides the desired decomposition.
\end{proof}

We will now express our equivalence theorem between subalgebras of $\cl(\ch)$ and their inclusions, and splitting maps on $\ch$ and their comprehension.

\begin{theorem}\label{equivalence of preorders}
The class of balanced splitting maps is the biggest class of splitting maps making $F/G$ an equivalence of preorders.
\end{theorem}

\begin{proof}
First we remind that $F$ and $G$ being preorder preserving is respectively the result of Proposition \ref{comprehension implies inclusion} and Proposition \ref{inclusion implies comprehension}.

Let us now prove that these two maps yield an equivalence of preorders. Let $\ca \in \vnalg(\ch)$, $F \circ G(\ca) = \stloc(\chi)$ with $\chi$ a canonical splitting map for $\ca$. It immediately follows that $F \circ G(\ca) = \ca$ and thus that $F \circ G(\ca) \subseteq \ca$ and $\ca \subseteq F \circ G(\ca)$. 

Conversely, let $\chi$ be a balanced splitting map. We then have that $G \circ F(\chi) = G(\stloc(\chi)) = \zeta$ is a canonical splitting map for $\ca = \stloc(\chi)$ and we want to prove that $\zeta \sqsubseteq \chi$ and $\chi \sqsubseteq \zeta$. First, let us Observe that it is sufficient to prove the result in the case where $\ca$ is factor. Indeed, by Lemma \ref{decomposition of split} we know that the $\chi_i = \chi \pi_i$ are such that $\Im(\chi_i) \subseteq \ch_{\LL}^{i} \otimes \ch_{\R}^{i}$ where the $\ch_{\LL}^{i}$ are pairwise orthogonal subspaces of $\chlc$ and the $\ch_{\R}^{i}$ pairwise orthogonal subspaces of $\chrc$. In the same way, we can obtain a decomposition of $\zeta$, as the sum of the $\zeta_i = \zeta \pi_i$ whose images have similar orthogonality properties as the ones of the $\chi_i$. Remark then that, for all $i$, the $(\chi_i)_{|\ch_i = \pi_i \ch}$ and $(\zeta_i)_{|\ch_i}$ are respectively balanced and canonical splitting maps for the factor algebra $\ca_i = \pi_i \ca$. Suppose now that, for all $i$, we have $\chi_i \sqsubseteq \zeta_i$, i.e.\ we have isometries $\bigcirc_i$ and $\tikzcircle{4pt}_i$ such that $(\bigcirc_i \otimes \id)\chi_i = (\id \otimes \tikzcircle{4pt}_i)\zeta_i$. Because the $\ch_{\LL}^{i}$ such that $\Im(\chi_i) \subseteq \ch_{\LL}^{i} \otimes \ch_{\R}^{i}$ are orthogonal (and because the same can be said for the lefty branch of the $\zeta_i$), choosing $\bigcirc = \oplus (\bigcirc_i)_{|\ch_i}$ and similarly for $\tikzcircle{4pt}$, we obtain that $(\bigcirc \otimes \id)\chi = (\id \otimes \tikzcircle{4pt})\zeta$ proving that $\chi \sqsubseteq \zeta$. A symmetric reasoning proves the implication for the reverse comprehension. It remains to show that, for all $i$, $\chi_i \sqsubseteq \zeta_i$ and $\zeta_i  \sqsubseteq \chi_i$, i.e.\ to prove the result in the factor case.

Suppose thus that $\ca$ is factor. By Lemma \ref{shape of good split factor}, there exists a basis $\{ \ket{lr} \}$ orthonormal families $\{ \ket{l+m} \}$ and $\{ \ket{r+m} \}$ such that $\chi(\ket{lr}) = \sum_m \lambda_m \ket{l+m}\ket{r+m}$. By construction of $\{ \ket{lr} \}$, and because $\zeta$ is a canonical splitting map associated to the same algebra $\ca$, we have that $\zeta(\ket{lr}) = \ket{\phi_l} \otimes \ket{\psi_r}$. Let's define $\bigcirc: \ket{l+m} \longrightarrow \ket{\phi_l} \otimes \ket{m}$ and $\tikzcircle{4pt} : \ket{\phi_r} \longrightarrow \sum_m \lambda_m \ket{m} \otimes \ket{m+r}$, we get that $(\bigcirc \otimes \id)\chi(\ket{lr}) = \sum_m \lambda _m \ket{\phi_l} \otimes \ket{m} \otimes \ket{m+r}=  \ket{\phi_l} \otimes ( \sum_m \lambda _m \ket{m} \otimes \ket{m+r}) = (\id \otimes \tikzcircle{4pt})\zeta(\ket{lr})$. A symmetric reasoning shows that $\chi \sqsubseteq \zeta$, proving that $F/G$ is then an equivalence. The fact that is the biggest class making it an equivalence is then the result of Remark \ref{equivalence implies balanced}.
\end{proof}

As a result, we have established a basic equivalence (as preorders) between the inclusion of von Neumann algebras and the comprehension of balanced splitting maps,
\begin{equation}
(\vnalg(\ch), \sqsubseteq) \cong (\balanced(\ch), \subseteq).
\end{equation}

After showing how to relate the algebraic view on quantum subsystems with our diagrammatic view, we have proved that the latter was able to fully capture the former by looking at the subclass of balanced splitting maps that form an equivalent theory to the one of general algebraic quantum subsystems. We will now make use of the diagrammatic approach to develop an understanding of the relationship between causality and localisability for general quantum subsystems. 

\section{Trace and non-signalling}

In the previous sections, we defined splitting maps in the context of pure quantum theory, as isometries splitting the space of states. In this section we will extend, in a very natural way, the notion of splitting map to isometric channels acting on density matrices. We will then define a notion of trace over a splitting map and show how it can be related to definition \ref{trace over an algebra} of the trace over a von Neumann algebra. We will then use these tools to prove a new result of causal decompositions of quantum channels that are non-signalling from an algebra to another.

\subsection{Mixed splitting maps and $\chi$-trace}

In the previous sections of this article, we had taken the point of view of pure quantum theory: the global system that we were trying to split was described by the Hilbert space $\ch$ and was given some additional structure through the splitting map. The question of extending splitting maps to acting on mixed quantum states follows very naturally and it can actually be answered quite simply, by doubling the isometry, just as it would be done for a unitary that we would like to see as a channel \cite{coecke_kissinger_2017}.
\begin{definition}[Mixed splitting map]
Given a pure splitting map $\chi : \ch \rightarrow \chlc \otimes \chrc$, we define the associated {\em mixed splitting map} $\chi : \cl(\ch) \rightarrow \cl( \chlc \otimes \chrc)$ as $\chi (\rho) = \chi \rho \chi^{\dagger}$. Diagrammatically, we will write it as

\begin{equation}
\input{figures/mixed_1.tikz}
\end{equation}
where the right part is an unfolding in pure quantum theory diagrams of the left, which is a mixed theory diagram.
\end{definition}

The natural follow-up is the question of the partial trace over a splitting map. The side of the splitting over which we will trace has to be chosen as convention; we decide to call trace over $\chi$, the trace over the right side of $\chi$, matching the definition of \cite{Arrighi2024quantumnetworks}.

\begin{definition}[$\chi$-trace]
Let $\chi : \ch \rightarrow \chlc \otimes \chrc$. The {\em partial trace over $\chi$}, which we will call {\em $\chi$-trace}, is the map $\Tr_{\chi} : \cl(\ch) \rightarrow \cl(\chlc)$ defined by $\Tr_{\chi}(\cdot) = \Tr_{\chrc}(\chi \cdot \chi^{\dagger})$. Diagrammatically we will write it

\begin{equation}
\input{figures/mixed_2.tikz} \, .
\end{equation}

\end{definition}

In particular, when $\chi$ is a lean splitting map, we obtain that the $\chi$-trace is equal (up to an isometry due to the choice of representation) to the trace (as it is defined in Definition \ref{trace over an algebra}) over the algebra $\stloc_{\R}(\chi)$; this strengthens the idea that splitting maps can capture perfectly the theory of quantum systems as von Neumann algebras.

\begin{theorem}\label{equivalence of traces}
Let $\cb \in \cl(\ch)$ be a von Neumann algebra and let $\chi : \ch \rightarrow \chlc \otimes \chrc$ be a lean splitting map such that $\stloc_{\R}(\chi) = \cb$. Then there exists an isometry $U$ such that $\Tr_{\chi}(\cdot) = U \Tr_{\cb}(\cdot) U^{\dagger}$.
\end{theorem}

\begin{proof}
Let $\cb \in \vnalg(\ch)$. First, let us prove this result when $\chi$ is canonical. By Theorem \ref{AW}, there exists a representation unitary $V: \ch \rightarrow \bigoplus_i (\ch_{{\LL}}^{i} \otimes \ch_{{\R}}^{i})$ for the pair of algebras $\cb'$ on the left and $\cb$ on the right. In the rest of this proof, we will identify $\ch$ and $ \bigoplus_i (\ch_{{\LL}}^{i} \otimes \ch_{{\R}}^{i})$, as it is only a matter of defining some structure on $\ch$ relative to the algebras $\cb'$ and $\cb$. The trace over $\cb$ is then defined as $\Tr_{\cb}(\cdot) = \bigoplus_i \Tr_{\ch_{\R}^{i}}(\pi_i \cdot \pi_i)$, where the $\pi_i$ are the projectors on the $\ch_{\LL}^{i} \otimes \ch_{\R}^{i}$, i.e.\ the atomic projectors of $\cb$.

Let now $\chi$ be a canonical splitting map such that $\stloc_{\R}(\chi) = \cb$. By definition there exists another representation unitary $\tilde{V} : \ch \rightarrow \bigoplus_i (\tilde{\ch}_{{\LL}}^{i} \otimes \tilde{\ch}_{{\R}}^{i})$, such that $\chi$ is simply the embedding $\ch \cong \bigoplus_i (\tilde{\ch}_{{\LL}}^{i} \otimes \tilde{\ch}_{{\R}}^{i}) \hookrightarrow (\oplus_i \tilde{\ch}_{{\LL}}^{i}) \otimes (\oplus_i \tilde{\ch}_{{\R}}^{i})$. We will call $\tilde{\ch}_{\LL} = \oplus_i \tilde{\ch}_{{\LL}}^{i}$, $\tilde{\ch}_{\R} = \oplus_i \tilde{\ch}_{{\R}}^{i}$ and $\tilde{\pi}_i $ the respective projectors on the $\tilde{\ch}_{\LL}^{i} \otimes \tilde{\ch}_{\R}^{i}$. Because $\tilde{V}$ is also a representation unitary for the pair of algebras $\cb'$ on the left and $\cb$ on the right, it is of the form $\tilde{V} = \bigoplus_i (U_{\LL}^{i} \otimes U_{\R}^{i})$ where the $U_{\LL}^{i} : \ch_{\LL}^{i} \rightarrow \tilde{\ch}_{\LL}^{i}$ are unitaries (and similarly on the right for the $U_{\R}^{i}$). We then have that, for all $\rho \in \cl(\ch)$, $\Tr_{\chi}(\rho) = \Tr_{\tilde{\ch}_{\R}}(\chi \rho \chi^{\dagger}) = \sum_{i,j} \Tr_{\tilde{\ch}_{\R}}(\chi \pi_i \rho \pi_j \chi^{\dagger}) = \sum_{i,j,k} \Tr_{\tilde{\ch}_{\R}^{k}}(\chi \pi_i \rho \pi_j \chi^{\dagger}) = \sum_{i,j,k} \Tr_{\tilde{\ch}_{\R}^{k}} ((U_{\LL}^{i} \otimes U_{\R}^{i}) \pi_i \rho \pi_j (U_{\LL}^{j} \otimes U_{\R}^{j})^{\dagger})  = \sum_i \Tr_{\tilde{\ch}_{\R}^{i}} ((U_{\LL}^{i} \otimes U_{\R}^{i}) \pi_i \rho \pi_i (U_{\LL}^{i} \otimes U_{\R}^{i})^{\dagger})  = \sum_i \Tr_{\ch_{\R}^{i}} ((U_{\LL}^{i} \otimes \id) \pi_i \rho \pi_i (U_{\LL}^{i} \otimes \id)^{\dagger}) = U \Tr(\rho) U^{\dagger}$ where $U = \bigoplus_i U_{\LL}^{i}$.

Let us now  go back to the general case and suppose that $\chi$ is any lean splitting map. Let $\ca = \stloc(\chi)$ and $\cb = \ca' = \stloc_{\R}(\chi)$. By Proposition \ref{decomposition of lean}, we know that there exist a canonical splitting map for $\ca$ and isometries $W_{\LL}$ and $W_{\R}$ such that $\chi = (W_{\LL} \otimes W_{\R}) \zeta$.  It follows that $\Tr_{\chi}(\cdot) = \Tr_{\chrc}(\chi \cdot \chi^{\dagger}) = \Tr_{\chrc}((W_{\LL} \otimes W_{\R}) \zeta \cdot \zeta^{\dagger}(W_{\LL}^{\dagger} \otimes W_{\R}^{\dagger})) = \Tr_{\chrz}((W_{\LL} \otimes \id) \zeta \cdot \zeta^{\dagger}(W_{\LL}^{\dagger} \otimes \id)) = W_{\LL} \Tr_{\chrz}(\zeta \cdot \zeta^{\dagger}) W_{\LL}^{\dagger} = W_{\LL} U \Tr_{\cb}(\cdot) U^{\dagger} W_{\LL}^{\dagger}$, where the last equality comes from applying this theorem that we've just proved to be true in the canonical case to $\zeta$.
\end{proof}

Let us quickly Observe that this theorem is not true in the case of general balanced splitting maps but only in the case of the lean ones. The idea is that if one of the $\phi_i$ is not separable in the decomposition of $\chi$ given by Corollary \ref{shape of good split factor}, then the $\chi$-trace, by tracing the right part of $\phi_i$, will induce some decoherence compared to the von Neumann algebra trace.

As we want to focus, in this last section, on signalling between algebras and causal decompositions, we will only consider splitting maps that behave like von Neumann algebras for their inclusion but also for their trace, i.e.\ lean splitting maps. Let us now express the notions of semi-causality and semi-locality in terms of splitting maps.

\begin{definition}[Heisenberg semi-causality]
We say that a unitary $U: \ch \rightarrow \ck$ is {\em Heisenberg semi-causal} from $\mathcal{A} \subseteq \cl(\ch)$ to $\mathcal{B} \subseteq \cl(\ck)$ if 
\begin{equation}  [U^{\dagger} \mathcal{B} U, \mathcal{A}] = 0 \, . \end{equation}
In terms of splitting maps, this is equivalent to asking that for any choice of lean splitting maps $\chi_{\ca}$, $\chi_{\ca'}$ and $\chi_{\cb}$ representing the corresponding algebras and for any $B \in \ch_{\LL}^{\chi_{\cb}}$, there exist $A_1 \in \ch_{\LL}^{\chi_{\ca'}}$ and $A_2 \in \ch_{\R}^{\chi_{\ca}}$ such that 
\begin{equation}
\input{figures/HSBtoSC_0.tikz} \, .
\end{equation}
\end{definition}

\begin{definition}[Schrödinger semi-causality]
We say that a channel $\ce: \cl(\ch) \rightarrow \cl(\ck)$ is {\em Schrödinger semi-causal} from $\mathcal{A} \subseteq \mathcal{L}(\ch)$ to $\mathcal{B} \subseteq \mathcal{L}(\ck)$ if there exists a completely positive and trace-preserving (CPTP) map $\ce'$ such that 

\begin{equation}\Tr_{\mathcal{B'}}[\ce(\cdot)] = \ce'(\Tr_{\mathcal{A}}[\cdot]) \, . \end{equation}
This condition can equivalently be phrased as the fact that for any choice of lean splitting maps $\chi_{\ca'}$ and $\chi_{\cb}$ there exists a CPTP map $\tilde{\ce}$ such that
\begin{equation}\label{semi-causality equation}
\input{figures/causal_1.tikz} \, .
\end{equation}

\end{definition}

In the particular case where $\ce$ is a unitary channel, i.e.\ of the form $\ce(\cdot) = U \cdot U^{\dagger}$ with $U : \ch \rightarrow \ck$ unitary, we obtain the following result:

\begin{theorem}
A unitary is Heisenberg semi-causal from $\ca$ to $\cb$ if and only if it is Schrödinger semi-causal from $\ca$ to $\cb$ . 
\end{theorem}
\begin{proof}
Let $U : \ch \rightarrow \ck$ be a unitary and let $\chi_{\ca'}$ and $\chi_{\cb}$ be lean splitting maps representing the corresponding algebras. Suppose that for all $B$, there exists $A$ such that

\begin{equation}
\input{figures/HSBtoSC_0.5.tikz} , .
\end{equation}
It follows that $\Tr_{\chi_{\ca'}}(\rho) = \Tr_{\chi_{\ca'}}(\sigma) \implies \Tr_{\chi_{\cb}}(U\rho U^\dagger) = \Tr_{\chi_{\cb}}(U \sigma U^\dagger)$. Indeed, let $\rho$ and $\sigma$ satisfying the aforementioned premise; then for all $B$,

\begin{equation}
\input{figures/HSBtoSC_1.tikz} \, ,
\end{equation}
and by removing the map and bending the wires back (which we can do because the equality is true for all $B$), we have proved that $\Tr_{\chi_{\cb}}(U\rho U^\dagger) = \Tr_{\chi_{\cb}}(U \sigma U^\dagger)$.

Recall now that, as $\chi_{\ca'}$ is lean, its image is of the form $\Im(\chi_{\ca'}) = \bigoplus_i \ch_{\LL}^{i} \otimes \ch_{\R}^{i} \subseteq \ch_{\LL}^{\chi_{\ca'}} \otimes \ch_{\R}^{\chi_{\ca'}} $. Let us construct the isometry $V : \bigoplus_i \ch_{\LL}^{i}  \rightarrow \ch$ such that, for all $i$, $ V_{|\ch_{\LL}^{i}} \ket{\psi^{i}_{\LL}} := \chi_{\ca'}^{\dagger}(\ket{\psi^{i}_{{\LL}}}\otimes\ket{0^{i}_{\R}})$ with $\ket{0^{i}_{\R}} \in \ch_{\R}^{i}$. Let $\{ \ket{k} \}_k$ be an orthonormal basis of $(\bigoplus_i \ch_{\LL}^{i})^{\perp}$ and choose some $\ket{\phi} \in \ch$. We define for all $k$ a map $V_k = \ketbra{\phi}{k}$. It follows that $\cf = V\cdot V^{\dagger} + \sum_k V_k \cdot V_k^{\dagger}$ is a quantum channel from $\cl(\ch_{\LL}^{\chi_{\ca'}})$ to $\cl(\ch)$ that has the property that $\Tr_{\chi_{\ca'}}(\cf(\Tr_{\chi_{\ca'}}(\rho))) = \Tr_{\chi_{\ca'}}(\rho)$. We finally define the channel $\tilde{\ce} : \rho \rightarrow \Tr_{\chi_{\cb}}(U \cf (\rho) U^{\dagger})$. We can compute that $\tilde{\ce}(\Tr_{\chi_{\ca'}}(\rho))  = \Tr_{\chi_{\cb}}(U \cf( \Tr_{\chi_{\ca'}}(\rho)) U^{\dagger}) = \Tr_{\chi_{\cb}}(U  \rho U^{\dagger})$ where the last equality comes from the fact that $\rho$ and $\cf(\Tr_{\chi_{\ca'}}(\rho))$ have the same $\chi_{\ca'}$-trace and thus the same $\chi_{\cb}$-trace after applying $U \cdot U^{\dagger}$.

Conversely, suppose that $U$ is Schrödinger semi-causal. Then for all $\rho$,

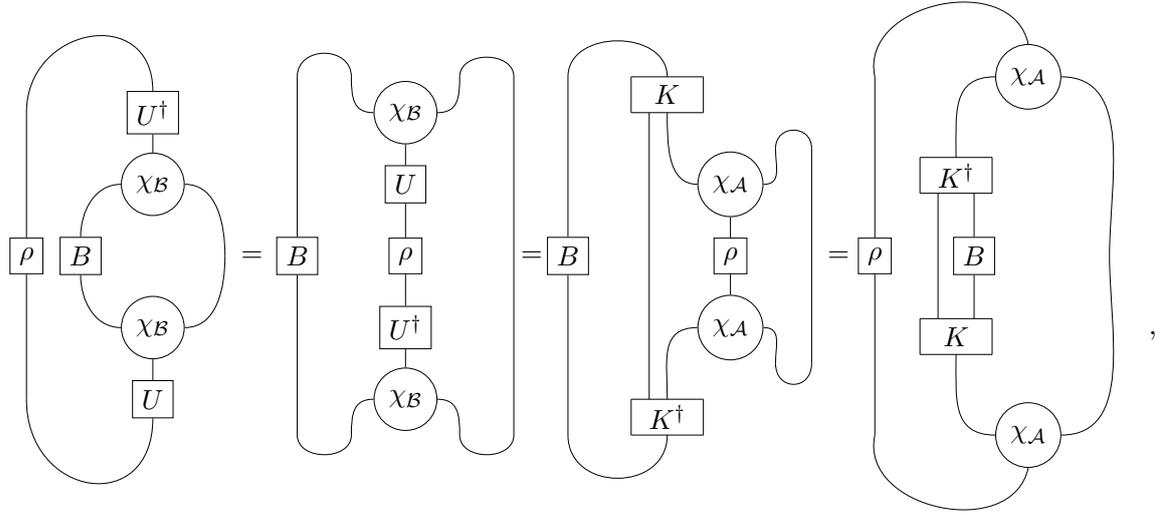
\begin{equation}
 \scalebox{0.95}{ \input{figures/SCtoHSB_1.tikz}} \, ,
\end{equation}
which again proves, by removing $\rho$ and bending the wires back, that

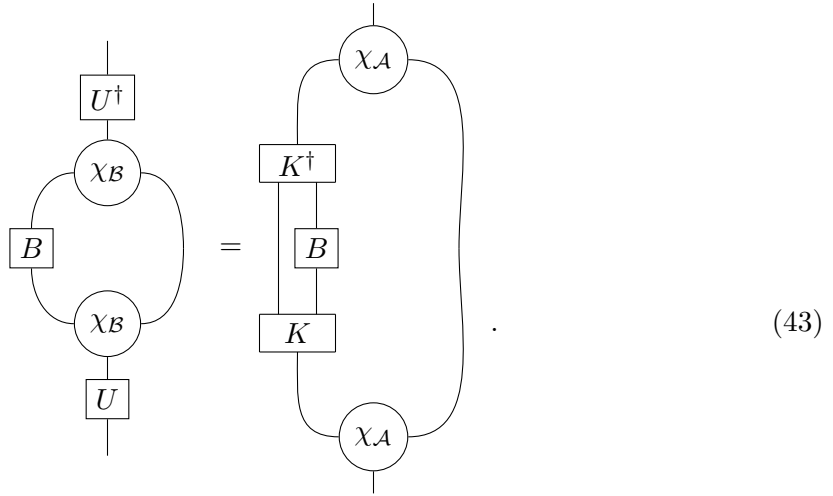
\begin{equation}
\input{figures/SCtoHSB_2.tikz} \, .
\end{equation}
As these implications are true for any choice of $\chi_{\ca'}$ and $\chi_{\cb}$, we have proved that $U$ is Heisenberg semi-causal if and only if it is Schrödinger semi-causal.
\end{proof}

This shows that the tracing-style definition, or equivalently the discarding-style definition, is not in some way incomplete and that it captures precisely the usual coherent notion of non-signalling in terms of commutation of operators when restricted to unitaries.

Let us now Observe that the Schrödinger causality, in its splitting maps phrasing, plays naturally with the notion of comprehension. Indeed, suppose that $\ce$ is semi-causal from $\ca$ to $\cb$ and let $\chi \sqsubseteq \chi_B$; then

\begin{equation}
\input{figures/causal_3.tikz} \, ,
\end{equation}
and a similar reasoning shows that if $\chi_{\ca'} \sqsubseteq \zeta$, then:

\begin{equation}
\input{figures/causal_4.tikz} \, .
\end{equation}

Observe that this implies that, to prove Schrödinger semi-causality, it is sufficient to prove equality  (\ref{semi-causality equation}) for one pair of corresponding lean splitting maps, as the mutual comprehension of lean splitting maps representing the same algebra will imply that the equality is true for any such pair. This computation also proves that even if the $\chi$-trace and the algebraic trace agree only in the case of lean splitting maps, semi-causality can actually be stated equivalently for any pair of balanced splitting maps respectively representing $\ca$ and $\cb$.  
Let us now give the definition of semi-localisability between a pair of splitting maps. We state it here in full generality but we will actually only use it in the rest of this section in the case of lean splitting maps.

\begin{definition}[Semi-localisability]
A channel $\ce : \cl(\ch) \longrightarrow \cl(\ck)$ is {\em semi-localisable} from $\chi \in \spl(\ch)$ to $\zeta \in \spl(\ck)$ if and only if there exist channels $\ce_1$ and $\ce_2$ such that it can be decomposed as

\begin{equation}
\input{figures/semi-localizable_1.tikz} \, ,
\end{equation}
with 

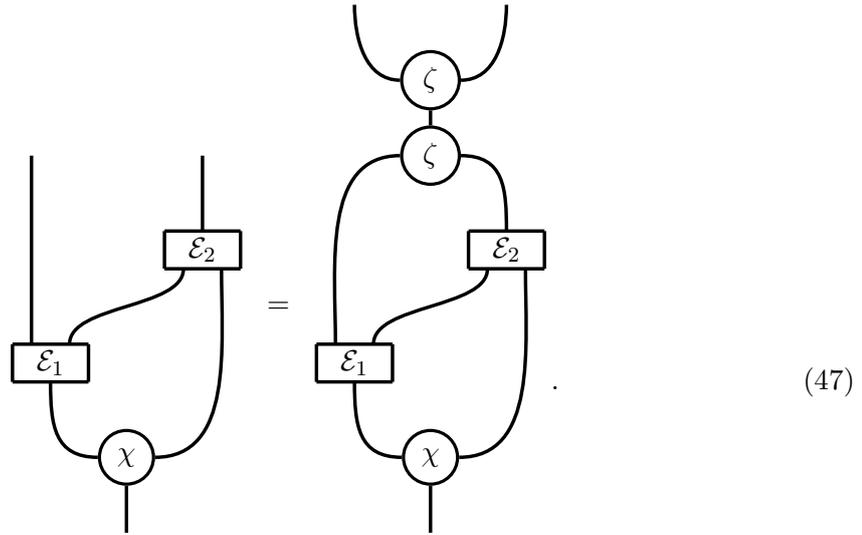
\begin{equation}\label{sl2}
\input{figures/semi-localizable_2.tikz} \, .
\end{equation}
\end{definition}
Note that if $\ch = \ck$, (\ref{sl2}) is equivalent to the $\zeta$-consistency of

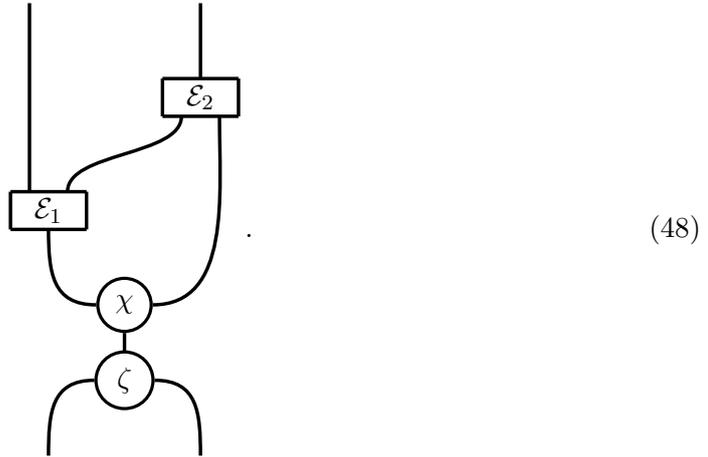
\begin{equation}
\input{figures/semi-localizable_3.tikz} \, .
\end{equation}

\begin{theorem}
A channel $\ce : \cl(\ch) \longrightarrow \cl(\ck)$ is semi-causal from $\ca \subseteq  \cl(\ch)$ to $\cb \subseteq \cl(\ck)$ if and only if there exist $\chi_{\ca'} \in \lean(\ch)$ and $\chi_{\cb} \in \lean(\ck)$, respectively representing $\ca'$ and $\cb$, such that $\ce$ is semi-localisable from $\chi_{\ca'}$ to $\chi_{\cb}$.
\end{theorem}

\begin{proof}
Suppose that the channel $\ce$ it is semi-localisable from $\chi_{\ca'}$ to $\chi_{\cb}$. Then immediately

\begin{equation}
\input{figures/sltosc.tikz} \, ,
\end{equation}
which proves that $\ce$ is semi-causal from $\ca$ to $\cb$. Conversely, suppose that $\ce$ is semi-causal; then for any choice of lean splitting maps $\chi_{\ca'}$ and $\chi_{\cb}$ there exists a channel $\tilde{\ce}$ such that 

\begin{equation}
\input{figures/causal_1.tikz} \, .
\end{equation}

By Stinespring dilation \cite{Stinespring:1955eig} of the channels $\ce$ and $\tilde{\ce}$ respectively, there exist Hilbert spaces $\ch_U$ and $\ch_V$ and isometries $U$ and $V$ such that 

\begin{equation}
\input{figures/sctosl_1.tikz} \, .
\end{equation}
Let us consider $\zeta_{\cb} = (\id \otimes W)\chi_{\cb}$ where $W : \ch_{\R}^{\chi_{\cb}} \rightarrow  \ch_{\R}^{\zeta_{\cb}}$ is an isometry such that $\dim(\ch_{\R}^{\zeta_{\cb}} \otimes \ch_U) \geq \dim(\ch_V \otimes \ch_{\R}^{\chi_{\ca'}})$. Observe that, by construction, $\zeta_{\cb}$ is also a lean splitting map representing the algebra $\cb$. We have that 

\begin{equation}
\input{figures/sctosl_2.tikz} \, .
\end{equation}
In particular this implies that 

\begin{equation}
\input{figures/sctosl_3.tikz}
\end{equation}
are two dilations of the same map. Because $\dim(\ch_{\R}^{\zeta_{\cb}} \otimes \ch_U) \geq \dim(\ch_V \otimes \ch_{\R}^{\chi_{\ca'}})$,  Lemma \ref{lemma dilation} tells us that there exists an isometry $T$ such that

\begin{equation}
\input{figures/sctosl_4.tikz}
\end{equation}

Then, because $U$ is a Stinespring dilation of $\ce$,

\begin{equation}
\input{figures/sctosl_5.tikz}
\end{equation}
and it directly follows that $\ce$ is semi-localisable from $\chi_{\ca'}$ to $\zeta_{\cb}$ as 

\begin{equation}
\input{figures/sctosl_6.tikz}
\end{equation}
and 

\begin{equation}
\input{figures/sctosl_7.tikz} \, .
\end{equation}

\end{proof}

\begin{remark}
The fact that the semi-causality of a channel $\ce$ is not equivalent to its semi-localisability with respect to any pair of lean splitting maps representing the algebras is only a problem of the dimensions of the Hilbert spaces involved in the representations not being compatible. Indeed when we construct $\zeta_{\cb}$, we have to ensure that $\dim(\ch_{\R}^{\zeta_{\cb}} \otimes \ch_U) \geq \dim(\ch_V \otimes \ch_{\R}^{\chi_{\ca'}})$ so that $T$ is an isometry from $\ch_V \otimes \ch_{\R}^{\chi_{\ca'}}$ to $\ch_{\R}^{\zeta_{\cb}} \otimes \ch_U$ and not the other way around.
\end{remark}

This is a causally faithful decomposition theorem for arbitrary finite-dimensional quantum subsystems. It shows both how splitting maps can be used to prove structural theorems regarding causality and locality, and even give a neat way to frame them to begin with.



\section{Conclusion}

In this article, we have defined the notion of splitting map as a tool to decompose a system into a pair of disjoint subsystems. We have shown that we were able to define some notion of (strict) locality with respect to a splitting map, through which we could link our diagrammatic approach to general algebraic quantum subsystems. We have equipped the set of splitting maps on a system with a preorder named comprehension and have then defined a few particular classes of splitting maps. The canonical splitting maps, given by the Artin-Wedderburn theorem, are constructive representations of algebraic subsystems. The balanced splitting maps form the biggest class that is equivalent (as a preordered set) to the theory of algebraic quantum subsystems. And the lean splitting maps (which form a subclass of the balanced ones) are able to diagrammatically capture the trace over the algebra they represent. Finally, we used this feature to obtain a causal decomposition for channels that are non-signalling from an algebra to another, generalising the result that was given in \cite{Eggeling_2002} in the factor case.



A natural follow-up is the question of how precisely this approach relates to that of \cite{Claeys_2024,allen2024cpinftybeyond2categoricaldilation,Reutter_2019}, which develops a compositional representation for non-factor subsystems using additional features and known graphical representations for $2$-categories. For instance, whilst there is a natural way to imagine constructing $2$-categories analogous to the $2$-category of $2$-Hilbert spaces \cite{baez1996higherdimensionalalgebraii2hilbert} from bimonoidal categories which add an additional monoidal structure analogous to the direct sum \cite{johnson2021bimonoidalcategoriesenmonoidalcategories}, the approach outlined in this paper suggests there could be a construction which achieves the same, but which simply leverages the basic notion of splitting internal to a monoidal category. In other words, we wonder if it could be possible to consider the identification of splitting maps and comprehension within monoidal categories as a new way to split an entire dagger monoidal category into a $2$-category.



The results presented in this paper appear to be a satisfactory approach to identifying non-factor subsystems within the language of symmetric monoidal categories, but possibly not a satisfactory diagrammatic approach. While splitting, locality, trace, and comprehension are simple interpretable diagrammatic equations or definitions, the same cannot be said for the representation that we have given for maps that are balanced or lean. In future work we intend to give a completely equational diagrammatic representation of those features too, in order to take a step closure to a fully diagrammatic identification of general quantum subsystems.

Finally, the traditional way to think of non-factor subsystems is in terms of decomposition; however, the dagger of the splitting map appears to represent an internal compositional representation in terms of a generalised tensor product of the kind explored in \cite{Arrighi2024quantumnetworks,arrighi2023generalisedtensorstraces}. For future work, it therefore seems promising to explore whether the established properties of splitting maps, in particular comprehension which mimics the associativity of composition for factor systems, really allow for a constructive/compositional way to reason with non-factor subsystems.

\section*{Acknowledgements}
It is a pleasure to thank Kathleen Barsse, Timothée Hoffreumon and James Hefford for helpful discussions and comments.

OM, MW and PA are partially funded by the European Union through the MSCA SE project QCOMICAL, by the French National Research Agency (ANR): projects TaQC ANR-22-CE47-0012 and within the framework of `Plan France 2030', under the research projects EPIQ ANR-22-PETQ-0007, OQULUS ANR-23-PETQ-0013, HQI-Acquisition ANR-22-PNCQ-0001 and HQI-R\&D
ANR-22-PNCQ-0002, and by the ID \#62312 grant from the John Templeton Foundation, as part of the \href{https://www.templeton.org/grant/the-quantum-information-structure-of-spacetime-qiss-second-phase}{‘The Quantum Information Structure of Spacetime’ Project (QISS)}.
MW was funded by the Engineering and Physical Sciences Research Council [grant number EP/W524335/1].
AV is supported by the STeP2 grant (ANR-22-EXES-0013) of Agence Nationale de la Recherche (ANR), the PEPR integrated project EPiQ (ANR-22-PETQ-0007) as part of Plan France 2030, the ANR grant TaQC (ANR-22-CE47-0012), and the ID \#62312 grant from the John Templeton Foundation, as part of the \href{https://www.templeton.org/grant/the-quantum-information-structure-of-spacetime-qiss-second-phase}{‘The Quantum Information Structure of Spacetime’ Project (QISS)}.
The opinions expressed in this publication are those of the authors and do not necessarily reflect the views of the John Templeton Foundation.

\bibliographystyle{quantum}
\bibliography{biblio}

@misc{baez1996higherdimensionalalgebraii2hilbert,
	archiveprefix = {arXiv},
	author = {John C. Baez},
	eprint = {q-alg/9609018},
	primaryclass = {q-alg},
	title = {Higher-Dimensional Algebra II: 2-Hilbert Spaces},
	url = {https://arxiv.org/abs/q-alg/9609018},
	year = {1996}}

@misc{johnson2021bimonoidalcategoriesenmonoidalcategories,
	archiveprefix = {arXiv},
	author = {Niles Johnson and Donald Yau},
	eprint = {2107.10526},
	primaryclass = {math.CT},
	title = {Bimonoidal Categories, $E_n$-Monoidal Categories, and Algebraic $K$-Theory},
	url = {https://arxiv.org/abs/2107.10526},
	year = {2021}}

@article{Reutter_2019,
	author = {Reutter, David J. and Vicary, Jamie},
	doi = {10.21136/hs.2019.04},
	journal = {Higher Structures},
	month = mar,
	number = {1},
	pages = {109--154},
	publisher = {Institute of Mathematics, Czech Academy of Sciences},
	title = {Biunitary constructions in quantum information},
	url = {http://dx.doi.org/10.21136/HS.2019.04},
	volume = {3},
	year = {2019}}

@article{Coecke_2018,
	author = {Coecke, Bob and Selby, John and Tull, Sean},
	doi = {10.4204/eptcs.266.7},
	issn = {2075-2180},
	journal = {Electronic Proceedings in Theoretical Computer Science},
	month = feb,
	pages = {104--118},
	publisher = {Open Publishing Association},
	title = {Two Roads to Classicality},
	url = {http://dx.doi.org/10.4204/EPTCS.266.7},
	volume = {266},
	year = {2018}}

@article{oreshkov_causal_order,
	author = {Oreshkov, Ognyan and Costa, Fabio and Brukner, {\v C}aslav},
	date = {2012/10/02},
	doi = {10.1038/ncomms2076},
	id = {Oreshkov2012},
	isbn = {2041-1723},
	journal = {Nature Communications},
	number = {1},
	pages = {1092},
	title = {Quantum correlations with no causal order},
	url = {https://doi.org/10.1038/ncomms2076},
	volume = {3},
	year = {2012}}

@article{PhysRevA.88.022318,
	author = {Chiribella, Giulio and D'Ariano, Giacomo Mauro and Perinotti, Paolo and Valiron, Benoit},
	doi = {10.1103/PhysRevA.88.022318},
	issue = {2},
	journal = {Phys. Rev. A},
	month = {Aug},
	numpages = {15},
	pages = {022318},
	publisher = {American Physical Society},
	title = {Quantum computations without definite causal structure},
	url = {https://link.aps.org/doi/10.1103/PhysRevA.88.022318},
	volume = {88},
	year = {2013}}

@article{Selby2021reconstructing,
	author = {Selby, John H. and Scandolo, Carlo Maria and Coecke, Bob},
	doi = {10.22331/q-2021-04-28-445},
	issn = {2521-327X},
	journal = {{Quantum}},
	month = apr,
	pages = {445},
	publisher = {{Verein zur F{\"{o}}rderung des Open Access Publizierens in den Quantenwissenschaften}},
	title = {Reconstructing quantum theory from diagrammatic postulates},
	url = {https://doi.org/10.22331/q-2021-04-28-445},
	volume = {5},
	year = {2021}}

@article{deutsch1989quantum,
	author = {Deutsch, David Elieser},
	doi = {10.1098/rspa.1989.0099},
	journal = {Proceedings of the royal society of London. A. mathematical and physical sciences},
	number = {1868},
	pages = {73--90},
	publisher = {The Royal Society London},
	title = {Quantum computational networks},
	volume = {425},
	year = {1989}}

@article{Vanrietvelde_2021,
	author = {Vanrietvelde, Augustin and Kristj{\'a}nsson, Hl{\'e}r and Barrett, Jonathan},
	doi = {10.22331/q-2021-07-13-503},
	issn = {2521-327X},
	journal = {Quantum},
	month = jul,
	pages = {503},
	publisher = {Verein zur Forderung des Open Access Publizierens in den Quantenwissenschaften},
	title = {Routed quantum circuits},
	url = {http://dx.doi.org/10.22331/q-2021-07-13-503},
	volume = {5},
	year = {2021}}

@misc{arrighi2023generalisedtensorstraces,
	archiveprefix = {arXiv},
	author = {Pablo Arrighi and Am{\'e}lia Durbec and Matt Wilson},
	eprint = {2202.11340},
	primaryclass = {quant-ph},
	title = {Generalised tensors and traces},
	url = {https://arxiv.org/abs/2202.11340},
	year = {2023}}

@misc{wilson2021composableconstraints,
	archiveprefix = {arXiv},
	author = {Matt Wilson and Augustin Vanrietvelde},
	eprint = {2112.06818},
	primaryclass = {math.CT},
	title = {Composable constraints},
	url = {https://arxiv.org/abs/2112.06818},
	year = {2021}}

@misc{vanrietvelde2025causaldecompositions1dquantum,
	archiveprefix = {arXiv},
	author = {Augustin Vanrietvelde and Octave Mestoudjian and Pablo Arrighi},
	eprint = {2506.22219},
	primaryclass = {quant-ph},
	title = {Causal Decompositions of 1D Quantum Cellular Automata},
	url = {https://arxiv.org/abs/2506.22219},
	year = {2025}}

@misc{vanrietvelde2025partitionsquantumtheory,
	archiveprefix = {arXiv},
	author = {Augustin Vanrietvelde and Octave Mestoudjian and Pablo Arrighi},
	eprint = {2506.22218},
	primaryclass = {quant-ph},
	title = {Partitions in quantum theory},
	url = {https://arxiv.org/abs/2506.22218},
	year = {2025}}

@article{schumacher_locality,
	author = {Schumacher, Benjamin and Westmoreland, Michael D.},
	doi = {10.1007/s11128-004-3193-y},
	journal = {Quantum Inf Process},
	pages = {13-34},
	title = {Locality and Information Transfer in Quantum Operations},
	volume = 4,
	year = 2005}

@article{barrett_gpts,
	author = {Barrett, Jonathan},
	doi = {10.1103/PhysRevA.75.032304},
	issue = {3},
	journal = {Phys. Rev. A},
	month = {Mar},
	numpages = {21},
	pages = {032304},
	publisher = {American Physical Society},
	title = {Information processing in generalized probabilistic theories},
	volume = {75},
	year = {2007}}

@article{Barrett2021,
	author = {Jonathan Barrett and Ralph Lorenz and Ognyan Oreshkov},
	doi = {10.1038/s41467-020-20456-x},
	journal = {Nature Communications},
	pages = {885},
	title = {Cyclic quantum causal models},
	url = {https://doi.org/10.1038/s41467-020-20456-x},
	volume = {12},
	year = {2021}}

@misc{vanrietvelde2023consistentcircuitsindefinitecausal,
	archiveprefix = {arXiv},
	author = {Augustin Vanrietvelde and Nick Ormrod and Hl{\'e}r Kristj{\'a}nsson and Jonathan Barrett},
	eprint = {2206.10042},
	primaryclass = {quant-ph},
	title = {Consistent circuits for indefinite causal order},
	url = {https://arxiv.org/abs/2206.10042},
	year = {2023}}

@article{beckman_localisable,
	author = {Beckman, David and Gottesman, Daniel and Nielsen, M. A. and Preskill, John},
	doi = {10.1103/PhysRevA.64.052309},
	issue = {5},
	journal = {Phys. Rev. A},
	pages = {052309},
	publisher = {American Physical Society},
	title = {Causal and localizable quantum operations},
	volume = {64},
	year = {2001}}

@article{lorenz_unitary,
	author = {Lorenz, Robin and Barrett, Jonathan},
	doi = {10.22331/q-2021-07-28-511},
	issn = {2521-327X},
	journal = {{Quantum}},
	pages = {511},
	publisher = {{Verein zur F{\"{o}}rderung des Open Access Publizierens in den Quantenwissenschaften}},
	title = {Causal and compositional structure of unitary transformations},
	volume = {5},
	year = {2021}}

@book{heunen_categories,
	author = {Heunen, Chris and Vicary, Jamie},
	doi = {10.1093/oso/9780198739623.001.0001},
	publisher = {Oxford University Press},
	title = {Categories for Quantum Theory: An Introduction},
	year = {2019}}

@article{gogioso_cpt,
	author = {Gogioso, Stefano and Scandolo, Carlo Maria},
	doi = {10.4204/EPTCS.266.23},
	journal = {EPTCS},
	pages = {367-385},
	title = {Categorical Probabilistic Theories},
	volume = {266},
	year = 2018}

@misc{gogioso_church,
	author = {Gogioso, Stefano},
	doi = {10.48550/arXiv.1905.13117},
	eprint = {1905.13117},
	primaryclass = {quant-ph},
	title = {A Process-Theoretic Church of the Larger Hilbert Space},
	year = {2019}}

@book{coecke_kissinger_2017,
	author = {Coecke, Bob and Kissinger, Aleks},
	doi = {10.1017/9781316219317},
	place = {Cambridge},
	publisher = {Cambridge University Press},
	title = {Picturing Quantum Processes: A First Course in Quantum Theory and Diagrammatic Reasoning},
	year = {2017}}

@article{coecke_cp,
	author = {Coecke, Bob and Heunen, Chris and Kissinger, Aleks},
	doi = {10.1007/s11128-014-0837-4},
	journal = {Quantum Inf. Process.},
	pages = {5179--5209},
	title = {Categories of quantum and classical channels},
	volume = 15,
	year = 2016}

@inproceedings{abramsky_coecke,
	author = {Abramsky, Samson and Coecke, Bob},
	booktitle = {Proceedings of the 19th Annual IEEE Symposium on Logic in Computer Science, 2004},
	doi = {10.1109/LICS.2004.1319636},
	pages = {415-425},
	title = {A categorical semantics of quantum protocols},
	year = 2004}

@article{chiribella_purification,
	author = {Chiribella, Giulio and D'Ariano, Giacomo Mauro and Perinotti, Paolo},
	doi = {10.1103/PhysRevA.81.062348},
	issue = {6},
	journal = {Phys. Rev. A},
	pages = {062348},
	publisher = {American Physical Society},
	title = {Probabilistic theories with purification},
	volume = {81},
	year = {2010}}

@article{Arrighi2024quantumnetworks,
	author = {Arrighi, Pablo and Durbec, Am{\'{e}}lia and Wilson, Matt},
	doi = {10.22331/q-2024-10-23-1508},
	issn = {2521-327X},
	journal = {{Quantum}},
	month = oct,
	pages = {1508},
	publisher = {{Verein zur F{\"{o}}rderung des Open Access Publizierens in den Quantenwissenschaften}},
	title = {Quantum networks theory},
	url = {https://doi.org/10.22331/q-2024-10-23-1508},
	volume = {8},
	year = {2024}}

@article{Chiribella_2018,
	author = {Chiribella, Giulio},
	doi = {10.3390/e20050358},
	issn = {1099-4300},
	journal = {Entropy},
	month = may,
	number = {5},
	pages = {358},
	publisher = {MDPI AG},
	title = {Agents, Subsystems, and the Conservation of Information},
	url = {http://dx.doi.org/10.3390/e20050358},
	volume = {20},
	year = {2018}}

@article{Eggeling_2002,
	author = {Eggeling, T and Schlingemann, D and Werner, R. F},
	doi = {10.1209/epl/i2002-00579-4},
	issn = {1286-4854},
	journal = {Europhysics Letters (EPL)},
	month = mar,
	number = {6},
	pages = {782--788},
	publisher = {IOP Publishing},
	title = {Semicausal operations are semilocalizable},
	url = {http://dx.doi.org/10.1209/epl/i2002-00579-4},
	volume = {57},
	year = {2002}}

@misc{mestoudjian2025selfcontainedproofartinwedderburntheorem,
	archiveprefix = {arXiv},
	author = {Octave Mestoudjian and Pablo Arrighi},
	eprint = {2507.08856},
	primaryclass = {math.RA},
	title = {A self-contained proof of the Artin-Wedderburn theorem in the case of finite-dimensional Von Neumann algebras},
	url = {https://arxiv.org/abs/2507.08856},
	year = {2025}}

@book{Nielsen_Chuang_2010,
	author = {Nielsen, Michael A. and Chuang, Isaac L.},
	doi = {10.1017/CBO9780511976667},
	place = {Cambridge},
	publisher = {Cambridge University Press},
	title = {Quantum Computation and Quantum Information: 10th Anniversary Edition},
	year = {2010}}

@book{Wilde_2017,
	author = {Wilde, Mark M.},
	doi = {10.1017/9781316809976},
	edition = {2},
	place = {Cambridge},
	publisher = {Cambridge University Press},
	title = {Quantum Information Theory},
	year = {2017}}

@article{Schumacher2004,
	archiveprefix = {arXiv},
	author = {Schumacher, B. and Werner, R. F.},
	eprint = {quant-ph/0405174},
	month = {5},
	title = {{Reversible Quantum Cellular Automata}},
	year = {2004}}

@inproceedings{arrighi2007,
	address = {Berlin, Heidelberg},
	archiveprefix = {arXiv},
	author = {Pablo Arrighi and Vincent Nesme and Reinhard Werner},
	booktitle = {Language and Automata Theory and Applications},
	doi = {10.1007/978-3-540-88282-4_8},
	eprint = {0711.3517},
	isbn = {978-3-540-88282-4},
	pages = {64--75},
	primaryclass = {quant-ph},
	publisher = {Springer Berlin Heidelberg},
	title = {One-dimensional quantum cellular automata over finite, unbounded configurations},
	year = {2008}}

@article{dariano2014,
	author = {D'Ariano, Giacomo and Manessi, Franco and Perinotti, P. and Tosini, Alessandro},
	doi = {10.1142/S0217751X14300257},
	journal = {International Journal of Modern Physics A},
	month = {07},
	title = {The Feynman problem and fermionic entanglement: Fermionic theory versus qubit theory},
	volume = {29},
	year = {2014}}

@inbook{Bianchi2023,
	address = {Singapore},
	author = {Bianchi, Eugenio and Livine, Etera R.},
	booktitle = {Handbook of Quantum Gravity},
	doi = {10.1007/978-981-19-3079-9_108-1},
	editor = {Bambi, Cosimo and Modesto, Leonardo and Shapiro, Ilya},
	isbn = {978-981-19-3079-9},
	pages = {1--29},
	publisher = {Springer Nature Singapore},
	title = {Loop Quantum Gravity and Quantum Information},
	url = {https://doi.org/10.1007/978-981-19-3079-9_108-1},
	year = {2023}}

@article{Eon2023relativistic,
	author = {Eon, Nathana{\"{e}}l and Molfetta, Giuseppe Di and Magnifico, Giuseppe and Arrighi, Pablo},
	doi = {10.22331/q-2023-11-08-1179},
	issn = {2521-327X},
	journal = {{Quantum}},
	month = nov,
	pages = {1179},
	publisher = {{Verein zur F{\"{o}}rderung des Open Access Publizierens in den Quantenwissenschaften}},
	title = {A relativistic discrete spacetime formulation of 3+1 {QED}},
	url = {https://doi.org/10.22331/q-2023-11-08-1179},
	volume = {7},
	year = {2023}}

@article{ahmad2022,
	author = {Ali Ahmad, Shadi and Galley, Thomas D. and H\"ohn, Philipp A. and Lock, Maximilian P. E. and Smith, Alexander R. H.},
	doi = {10.1103/PhysRevLett.128.170401},
	issue = {17},
	journal = {Phys. Rev. Lett.},
	month = {Apr},
	numpages = {8},
	pages = {170401},
	publisher = {American Physical Society},
	title = {Quantum Relativity of Subsystems},
	url = {https://link.aps.org/doi/10.1103/PhysRevLett.128.170401},
	volume = {128},
	year = {2022}}

@article{relativeSubsystems,
	author = {Castro-Ruiz, Esteban and Oreshkov, Ognyan},
	doi = {10.1038/s42005-025-02036-x},
	journal = {Communications Physics},
	month = {04},
	title = {Relative subsystems and quantum reference frame transformations},
	volume = {8},
	year = {2025}}

@article{FeynmanQCA,
	author = {Feynman,, R. P.},
	doi = {10.1007/BF01886518},
	journal = {Foundations of Physics (Historical Archive)},
	number = {6},
	pages = {507--531},
	publisher = {Springer},
	title = {{Quantum mechanical computers}},
	volume = {16},
	year = {1986}}

@article{Giacomini2017,
	archiveprefix = {arXiv},
	author = {Giacomini, Flaminia and Castro-Ruiz, Esteban and Brukner, {\v{C}}aslav},
	doi = {10.1038/s41467-018-08155-0},
	eprint = {1712.07207},
	journal = {Nature Communications},
	number = {1},
	pages = {494},
	primaryclass = {quant-ph},
	title = {{Quantum mechanics and the covariance of physical laws in quantum reference frames}},
	volume = {10},
	year = {2019}}

@article{Bartlett2007,
	archiveprefix = {arXiv},
	author = {Bartlett, Stephen D. and Rudolph, Terry and Spekkens, Robert W.},
	doi = {10.1103/RevModPhys.79.555},
	eprint = {quant-ph/0610030},
	journal = {Review of Modern Physics},
	pages = {555--609},
	title = {{Reference frames, superselection rules, and quantum information}},
	volume = {79},
	year = {2007}}

@book{farenick,
	author = {Farenic, Douglas R.},
	doi = {10.1007/978-1-4613-0097-7},
	publisher = {Springer New York},
	title = {Algebras of Linear Transformations},
	year = 2001}

@article{viola2001,
	archiveprefix = {arXiv},
	author = {{Viola}, Lorenza and {Knill}, Emanuel and {Laflamme}, Raymond},
	doi = {10.1088/0305-4470/34/35/331},
	eprint = {quant-ph/0101090},
	journal = {Journal of Physics A Mathematical General},
	month = sep,
	number = {35},
	pages = {7067-7079},
	primaryclass = {quant-ph},
	title = {{Constructing qubits in physical systems}},
	volume = {34},
	year = 2001}

@article{zanardi2001,
	archiveprefix = {arXiv},
	author = {{Zanardi}, Paolo},
	doi = {10.1103/PhysRevLett.87.077901},
	eid = {077901},
	eprint = {quant-ph/0103030},
	journal = {Physical Review Letters},
	month = aug,
	number = {7},
	pages = {077901},
	primaryclass = {quant-ph},
	title = {{Virtual Quantum Subsystems}},
	volume = {87},
	year = 2001}

@article{zanardi2003,
	archiveprefix = {arXiv},
	author = {{Zanardi}, Paolo and {Lidar}, Daniel A. and {Lloyd}, Seth},
	doi = {10.1103/PhysRevLett.92.060402},
	eid = {060402},
	eprint = {quant-ph/0308043},
	journal = {Physical Review Letters},
	month = feb,
	number = {6},
	pages = {060402},
	primaryclass = {quant-ph},
	title = {{Quantum Tensor Product Structures are Observable Induced}},
	volume = {92},
	year = 2004}

@article{Stinespring:1955eig,
	author = {Stinespring, W. Forrest},
	doi = {10.1090/s0002-9939-1955-0069403-4},
	journal = {Proc. Am. Math. Soc.},
	number = {2},
	pages = {211--216},
	title = {{Positive functions on {\ensuremath{\mathit{C}}}*-algebras}},
	volume = {6},
	year = {1955}}

@book{Maclane1971-MACCFT,
	author = {Saunders Maclane},
	doi = {10.1007/978-1-4757-4721-8},
	publisher = {Springer},
	title = {Categories for the Working Mathematician},
	year = {1971}}

@article{Chardonnet_2025,
	author = {Chardonnet, Kostia and de Visme, Marc and Valiron, Beno{\^\i}t and Vilmart, Renaud},
	doi = {10.46298/lmcs-21(2:13)2025},
	issn = {1860-5974},
	journal = {Logical Methods in Computer Science},
	month = may,
	publisher = {Centre pour la Communication Scientifique Directe (CCSD)},
	title = {The Many-Worlds Calculus},
	url = {http://dx.doi.org/10.46298/lmcs-21(2:13)2025},
	volume = {Volume 21, Issue 2},
	year = {2025}}

@article{Vicary_2010,
	author = {Vicary, Jamie},
	doi = {10.1007/s00220-010-1138-0},
	issn = {1432-0916},
	journal = {Communications in Mathematical Physics},
	month = nov,
	number = {3},
	pages = {765--796},
	publisher = {Springer Science and Business Media LLC},
	title = {Categorical Formulation of Finite-Dimensional Quantum Algebras},
	url = {http://dx.doi.org/10.1007/s00220-010-1138-0},
	volume = {304},
	year = {2010}}

@misc{allen2024cpinftybeyond2categoricaldilation,
	archiveprefix = {arXiv},
	author = {Robert Allen and Dominic Verdon},
	eprint = {2310.15776},
	primaryclass = {math.OA},
	title = {CP$^{\infty}$ and beyond: 2-categorical dilation theory},
	url = {https://arxiv.org/abs/2310.15776},
	year = {2024}}

@article{Claeys_2024,
	author = {Claeys, Pieter W and Lamacraft, Austen and Vicary, Jamie},
	doi = {10.1088/1751-8121/ad653f},
	journal = {Journal of Physics A: Mathematical and Theoretical},
	month = {jul},
	number = {33},
	pages = {335301},
	publisher = {IOP Publishing},
	title = {From dual-unitary to biunitary: a 2-categorical model for exactly-solvable many-body quantum dynamics},
	url = {https://doi.org/10.1088/1751-8121/ad653f},
	volume = {57},
	year = {2024}}

@book{takesaki2001theory,
  title={Theory of Operator Algebras I},
  author={Takesaki, M.},
  isbn={9783540422488},
  lccn={79013655},
  series={Encyclopaedia of Mathematical Sciences},
  doi = {10.1007/978-1-4612-6188-9},
  year={2001},
  publisher={Springer Berlin Heidelberg}
}

@article{Kramer_2018,
   title={Operational locality in global theories},
   volume={376},
   ISSN={1471-2962},
   url={http://dx.doi.org/10.1098/rsta.2017.0321},
   DOI={10.1098/rsta.2017.0321},
   number={2123},
   journal={Philosophical Transactions of the Royal Society A: Mathematical, Physical and Engineering Sciences},
   publisher={The Royal Society},
   author={Krämer, Lea and del Rio, Lídia},
   year={2018},
   month=may, pages={20170321} }

@article{Yashin_2020,
   title={Properties of operator systems, corresponding to channels},
   volume={19},
   ISSN={1573-1332},
   url={http://dx.doi.org/10.1007/s11128-020-02693-7},
   DOI={10.1007/s11128-020-02693-7},
   number={7},
   journal={Quantum Information Processing},
   publisher={Springer Science and Business Media LLC},
   author={Yashin, V. I.},
   year={2020},
   month=may }

\appendix

\section{Algebra results}

We remind a result on the shape of von Neumann algebra homomorphisms \cite{vanrietvelde2025causaldecompositions1dquantum}.

\begin{proposition} \label{prop: kernels are blocks}
Let $f: \ca \subseteq \cl(\ch_{\ca}) \to \cb \subseteq \cl(\ch_{\cb})$ be a homomorphism of von Neumann algebras and let $\Atproj(\cz(\ca)) = \{ \pi_i\}_{i \in I}$. Let us define $I_f := \{i \in I | f(\pi_i) \neq 0\}$, then the orthogonal projectors $\mu := \sum_{i \in I_f} \pi_i$ and $\mub := \id - \mu = \sum_{i \in I \setminus I_f} \pi^i$ are such that:
    
    \begin{subequations}
        \be \label{eq: kernel 1} \ca = \mu \ca \oplus \mub \ca \, ; \ee
        \be \label{eq: kernel 2}\ker f = \mub \ca \, ; \ee
        \be \label{eq: kernel 3} f|_{\mu \ca} \textrm{ is injective.} \ee
    \end{subequations}
\end{proposition}

\begin{proof}
The proof of \ref{eq: kernel 1} is immediate by writing the identity operator as the sum of the $\pi_i$ and splitting it into two parts, for $\mu$ and $\mub$ respectively. By definition of $I_f$ and $\mu$, the reverse inclusion of \ref{eq: kernel 2} is also immediate. Let us now prove that $\ker f \subseteq \mub \ca $.

Suppose that it is not the case and that there exists an $A \in \ker f \setminus \mub \ca$. Because $\mub \ca \subseteq \ker f$, we have that $f(\mu A) = f(\mu A) + f(\mub A) = f(A) = 0 $ and thus without loss of generality we can suppose that $0 \neq A \in \mu \ca \cap \ker f$.

It follows that $A A^{\dagger}$ is a non-zero positive element of $\mu \ca$ which is isomorphic to a $\bigoplus_{i \in I_f} \cl(\ch_{\ca}^{i})$ by theorem \ref{blocks are factor}. By spectral theorem $A A^{\dagger}$ can be decomposed as $A A^{\dagger} = \sum_k \alpha_k \nu_k$ where the $\nu_k$ are non-zero orthogonal and pairwise orthogonal rank-$1$ projectors of $\mu \ca$ and the $\alpha_k$ are positive real numbers. Because $f$ is an homomorphism of von Neumann algebras, the $f(\nu_k)$ inherit the properties of the $\nu_k$ and are also a family of orthogonal and pairwise orthogonal rank-$1$ projectors. We can then compute that
\be 
0 = f(A)f(A^{\dagger}) = f(A A^{\dagger}) = f (\sum_k \alpha_k \nu_k) = \sum_k \alpha_k f(\nu_k)
\ee
which implies that for all $k$, if $\alpha_k \neq 0$, then $f(\nu_k) = 0$. And because $A \neq 0$, at least one of the $f(\nu_k)$ is null. Without loss of generality let us suppose that it is $f(\nu_0)$. As $\nu_0 \in \mu \ca$, $\nu_0 = \mu \nu_0 = \sum_{i \in I_f} \pi_i \nu_0$ and because $\nu_0$ is of rank $1$ there exists a unique $i \in I_f$ such that $\nu_0 \in \pi_i \ca$.

Identifying $\pi_i \ca$ with its representation $\cl(\ch_{\ca}^{i})$ given by theorem \ref{blocks are factor}, we can chose an orthonormal basis $\{ \ket{j} \}$ of $\ch_{\ca}^{i}$ such that $\nu_0 = \ketbra{0}{0}$. Observe that $\pi_i \ca$ is then spanned by the set $\{ \ketbra{m}{n} \}$ but also that 
\be
f(\ketbra{m}{n}) = f(\ketbra{m}{0})f(\ketbra{0}{0})f(\ketbra{0}{n}) =  f(\ketbra{m}{0})0f(\ketbra{0}{n}) = 0
\ee
It follows that $\pi_i \ca \subseteq \ker h$ and thus that $\pi_i \in \ker h$, which is not possible as $i \in I_f$. This proves that $\ker f \setminus \mub \ca = \emptyset$ and that $\ker f = \mub \ca$. It also follows that $\ker f|_{\mu \ca} = \mu \ca \cap \ker f  = \mu \ca \cap \mub \ca = \{ 0 \}$ proving \ref{eq: kernel 3}.
\end{proof}

\section{Preorder preservation of the AW representation}\label{bigproof}

We now give the proof of the following proposition, stating that any choice of a canonical splitting map for every $\ca \in \vnalg(\ch)$ will be a preorder preserving map from $\vnalg(\ch)$ to $\spl(\ch)$.

\begin{proposition}
Let $\ca_S \subseteq \ca_B$ be two von Neumann subalgebras of $\cl(\ch)$. Let $\chi : \ch \rightarrow \ch_{\LL}^B \otimes \ch_{\R}^B$ and $\zeta : \ch \rightarrow \ch_{\LL}^S \otimes \ch_{\R}^S$ be canonical splitting maps for $\ca_B$ and $\ca_S$ respectively. Then $\zeta \sqsubseteq \chi$.
\end{proposition}

\begin{proof}

The proof of this proposition is based on the constructive proof of Artin-Wedderburn theorem (in the case of finite dimensional von Neumann algebras) given in \cite{mestoudjian2025selfcontainedproofartinwedderburntheorem}. We will first explain, without going into full details, how these representation unitaries can be constructed and then prove the comprehension of the splitting maps.

Let $P=\{P_i\}_{i=1\ldots p}$ be a family of non-zero orthogonal and pairwise orthogonal projectors of $\mathcal{A}$, i.e. \\
$(i)$ $\forall i=1\ldots p \; 0 \neq P_i = P_i^{\dagger} \in\mathcal{A}$; \\
$(ii)$ $\forall i,j=1\ldots p \; P_iP_j=\delta_{ij}P_i\in\mathcal{A}$. \\
Notice that such a family always exists as $\{ \id \} \subseteq \mathcal{A}$ is one.\\
Moreover we will take $P$ to be a maximal among such families, i.e. so that there is no set $Q=\{Q_i\}$ verifying conditions $(i),(ii)$ and such that $\mathcal{P}\subset \mathcal{Q}$, with
$\mathcal{P},\mathcal{Q}$ the algebraic closures of $P,Q$. As the $P_i$ are orthogonal projectors, we can understand maximality as saying that no $P_i$ can be decomposed as the sum of orthogonal projectors of $\mathcal{A}$ that would be orthogonal one to the other and to all the other $P_j$.
Finally note that:\\
$(iii)$ $\sum_{i=1\ldots p} P_i = \id$,\\
otherwise $\id-\sum_{i=1\ldots p} P_i$ may be added to the set.\vspace{1mm}\\

Such a family has the following properties:

\begin{itemize}
\item The following relation $P_i R P_j \Leftrightarrow P_i \ca P_j \neq 0$ is an equivalence relation. Projectors of the same class $I$ are in the same block $\cl(\ch_{\LL}^{I}) \otimes \id_{\ch_{\R}^{I}}$ of the Artin-Wedderburn decomposition. Different blocks are in direct sum.
\item  For all $i$, $P_i\mathcal{A}P_i=\mathbb{C}P_i$. This comes from the fact that the family is maximal. The $P_i$ indicate redudancies that will be seen as tensoring with the identity in a block $\cl(\ch_{\LL}^{I}) \otimes \id_{\ch_{\R}^{I}}$ of the decomposition.
\item Projectors of the same class have the same rank
\end{itemize}

Because the $P_i$ of a same class $I$ form an orthogonal, complete set of projectors of equal dimension it is possible to find a unitary $U_{I}$ that respectively sends them to $\ketbra{i}{i} \otimes \id_{\ch_{\R}^{I}}$ (where $\ch_{\R}^{I}$ is of dimension $q$ their common rank). This $U_I$ is the data of a mapping from $\Im(P_i)$ into $\ch_{{\R}_I}$ for for each $i \in I$ but nothing forces yet the different mappings to match with respect to what can be done in $\ca$. This is why in the end the unitary for a block (or equivalence class) is the following

\begin{equation}
    W_I = \sum_i (\ketbra{i}{1_I} \otimes \id_{\ch_{\R}^{I}})UP_{1_I}A P_i
\end{equation}

and the general mapping is the sum of these mapping for each blocks. One can easily show that any representation unitary $W$ is of this form. It suffices to take $P_i$ as the preimage of $ \ketbra{i}{i} \otimes \id_{\ch_{\R}^{I}}$ and U as the mapping on the left part of the tensor. 

With words, one could say that a splitting map of this form splits the space into the direct sum of $\Im(P_i)$'s, and then sends to the left which $\Im(P_i)$ you are in, and to the right how you are mapped into a space $\ch_{{\LL}_I}$ that is the same for all $P_i$ of a same equivalence class $I$. Let's also add that the mapping to the right is coherent with what the algebra $\ca$ can do, i.e. an element of $\Im(P_i)$ and an element of $\Im(P_k)$ that are related by an operation of the considered algebra should have the same mapping in $\ch_{{\LL}_I}$ up to a scalar.

Let us first show that the choice of the maximal family of non-zero, orthogonal and pairwise orthogonal projectors for a fixed subalgebra $\ca \subseteq \cl(\ch)$ doesn't matter. Indeed, we will prove in Lemma \ref{p2p} that any two such families are equal up to unitary and in Lemma \ref{change of chi} that a canonical splitting map for $\ca$ can be written $\chi = \sum_I \sum_i (\ketbra{i}{1_I} \otimes \id_{\ch_{\R}^{I}})UP_{1_I}A P_i$ no matter the choice of such a family of projectors $\{ P_i \}$.

\begin{lemma}\label{p2p}
Let $\ca$ be a subalgebra of $\cl(\ch)$. Let $\{P_i\}$ and $\{Q_j\}$ be two maximal families of non-zero, orthogonal and pairwise orthogonal projectors of $\ca$. Then there exists a unitary $U$ in $\ca$ such that, up to reordering the projectors, $P_i = U Q_i U^{\dagger}$.
\end{lemma}

\begin{proof}
Let a $\{P_i\}$ and $\{Q_j\}$ be two maximal families of non-zero, orthogonal and pairwise orthogonal projectors of $\ca$. We have seen in (the proof of) Theorem \ref{AW} that one could construct a unitary $U$ such that $U\mathcal{A} U^\dagger = \bigoplus_I \cl(\ch_{\LL}^{I}) \otimes \id_{\ch_{\R}^{I}}$ and that for a $P_i$ in the equivalence class $I$, $U P_i U^{\dagger} = \ketbra{i}{i} \otimes \id_{\ch_{\R}^{I}}$. Then, because the $Q_j$ are elements of $\ca$ and $U$ is unitary, we have that the $U Q_j U^{\dagger} = \sum_I Q_j^{I} \otimes \id_{\ch_{\R}^{I}}$ are a maximal family of non-zero orthogonal and pairwise orthogonal projectors of $\bigoplus_I \cl(\ch_{\LL}^{I}) \otimes \id_{\ch_{\R}^{I}}$. Because the sum is orthogonal, properties of the $Q_j$ transfer to the $Q_j^{I}$ which we can deduce are also orthogonal projectors. Moreover this also induces that for any two projectors $Q_{j_1}$ and $Q_{j_2}$ and any equivalence class $I$, $Q_{j_1}^{I}$ and $Q_{j_2}^{I}$ are orthogonal to each other. Non-zeroness tells us that for a fixed $j$, at least one of the $Q_j^{I}$ is non zero, and maximality tells us that at most one is non-zero, otherwise we could replace $U Q_j U^{\dagger}$ in the family by the different non-zero $Q_j^{I} \otimes \id_{\ch_{\R}^{I}}$ appearing in its decomposition and which are orthogonal projectors, orthogonal to each other and also to all other $U Q_k U^{\dagger}$ as we have seen that pairwise orthogonality is true in each block. Finally maximality also forces that $Q_j^{I}$ is a rank 1 projector which we can then write $\ketbra{\tilde{j}}{\tilde{j}}$. In particular this tells us that there is the same number of $P_i$ and $Q_j$ in each block and that they are all projectors of the same rank; and we can thus identify the $i$ indices with the $j$ ones. Let $V_I$ be for each block $I$ the unitary such that $V_I\ket{i} = \ket{\tilde{i}}$. Then  $U^{\dagger}(\sum_I V_I\otimes \id_{\ch_{\R}^{I}}) U P_i U^{\dagger}(\sum_I V_I^{\dagger} \otimes \id_{\ch_{\R}^{I}})U = Q_i$. Note that $(\sum_I V_I^{\dagger} \otimes \id_{\ch_{\R}^{I}})$ is a unitary of $\bigoplus_I \cl(\ch_{\LL}^{I}) \otimes \id_{\ch_{\R}^{I}}$ and thus that $U^{\dagger}(\sum_I V_I^{\dagger} \otimes \id_{\ch_{\R}^{I}})U$ is a unitary of $\ca$.
\end{proof}

It follows that a canonical splitting map for $\ca$ can be decomposed with respect to any such family of non-zero, orthogonal and pairwise orthogonal projectors of $\ca$.

\begin{lemma}
\label{change of chi}
Let $\ca$ be a subalgebra of $\cl(\ch)$ and $\chi$ a canonical splitting map for $\ca$. Let $\{Q_j\}$ be any maximal family of non-zero, orthogonal and pairwise orthogonal projectors of $\ca$. Then there exists a unitary $U$, an element $A$ of $\ca$ and a basis $\{ \ket{i} \}$ of $\chlc$ such that 
\begin{equation}
\chi = \sum_I \sum_i (\ketbra{i}{1_I} \otimes \id_{\ch_{\R}^{I}})UQ_{1_I}AQ_i
\end{equation}
\end{lemma}

\begin{proof}
As seen previously, $\chi$ can be written as $ \chi = \sum_I \sum_i (\ketbra{\tilde{i}}{1_I} \otimes \id_{\ch_{\R}^{I}})\tilde{U}P_{1_I}\tilde{B} P_i$ with $\tilde{U}$ a unitary, $\tilde{A}$ an element of $\ca$ and $\{ \ket{\tilde{i}} \}$ a basis of $\chlc$. By lemma \ref{p2p} we know that there exists a unitary $V$ in $\ca$ such that $P_i = V Q_i V^{\dagger}$. It follows that $\chi = \sum_I \sum_i (\ketbra{\tilde{i}}{1_I} \otimes \id_{\ch_{\R}^{I}}) \tilde{U} V Q_{1_I} V^{\dagger}\tilde{A}VQ_iV^{\dagger} $. By defining $ U = \tilde{U}V$ and $A = V^{\dagger}\tilde{A}V$ (which is indeed such that for all $i \in I$, $Q_{1_I} A Q_i \neq 0$ and in particular such that $Q_{1_I} A Q_i A^{\dagger} Q_{1_I} = Q_{1_I}$ because $\tilde{A}$ has this property with respect to the $P_i$), we get that $\chi = \sum_I \sum_i (\ketbra{\tilde{i}}{1_I} \otimes \id_{\ch_{\R}^{I}})UQ_{1_I}A Q_i V^{\dagger}$. Moreover, as $V \in \ca = \stloc(\chi)$, thus there exists $V_{\LL} \in \cons(\chi)$ such that $(V_{\LL} \otimes \id_{\chrc})  \chi = \chi V$ (and  $(V_{\LL}^{\dagger} \otimes \id_{\chrc}) \chi = \chi V^{\dagger}$). It follows that $\chi = \chi  V^{\dagger} V = (V_{\LL}^{\dagger} \otimes \id_{\chrc}) \chi V$ and by defining $\ket{i} = V_{\LL} \ket{\tilde{i}}$ we get that $\chi = \sum_I \sum_i (\ketbra{i}{1_I} \otimes \id_{\ch_{\R}^{I}})UQ_{1_I}A Q_i V^{\dagger}$ which concludes the proof.
\end{proof}

Let us now prove Proposition \ref{inclusion implies comprehension}. We have that $\zeta =\sum_J \sum_j (\ketbra{j}{1_J} \otimes \id_{\ch_{\R}^{J}})VQ_{1_J}B Q_j$  with the $\{Q_j\}$ a maximal family of non-zero, orthogonal and pairwise orthogonal projectors of $\ca_S$. Let us define the right family of such projectors of $\ca_B$ that will make our computation possible. If the family $\{ Q_j \}$ is also maximal for $\ca_B$ (which doesn't necessarily mean that $\ca_B = \ca_S$, as we can have $Q_j \ca_B Q_k \neq Q_j \ca_S Q_k$), we define $\{P_i\} = \{Q_j\}$. Otherwise, if $\{Q_j\}$ is not maximal for $\ca_B$, there exists a maximal family $\{P_i\}$ of such projectors for $\ca_B$ such that the $P_i$ generates the $Q_j$. Because the $P_i$ are orthogonal, pairwise orthogonal and because the $Q_j$ also form a family of projectors of $\ca_B$ with the same properties, just not a maximal one, it follows that each of the $Q_j$ can be written as a sum of some of the $P_i$, all with coefficient $1$; such that each $P_i$ appears in the decomposition of precisely one $Q_j$. This means that we can reindex the $P_i$ as $P_i^j$ (for a fixed $j$, the $i$ index varies from $1$ to a number $n_j$ and indicates which projector of the decomposition we are considering). However this $\{P_i^j\}$ family might not have all the properties we need thus we will build a new one $\{\tilde{P}_i^j\}$.

Let us work in one equivalence class $J$ of the $Q_j$ (as projectors of $\ca_S$). Consider $Q_1 = \sum_i P_i^1$, a projector of the equivalence class; we define $\{ \tilde{P}_i^1 \}=  \{P_i^1 \}$. Let us now fix a $j \in J$; by definition of the equivalence class there exists $A \in \ca_S$ such that $Q_1 A Q_j \neq 0$. If we multiply $A$ by the right scalar we can even have that $Q_j A^{\dagger} Q_1 A Q_j = Q_j$ which means that we can use the decomposition of $Q_1$ to get a new decomposition of $Q_j$ as $Q_j = \sum_i Q_j A^{\dagger} P_i^1 A Q_j$. And we define $\tilde{P}_i^j =  Q_j A^{\dagger} P_i^1 A Q_j$. Doing that for each $j \in J$ and then for each equivalence class $J$ defines a new family $\{\tilde{P}_i^j\}$ of elements of $\ca_B$. Morally we are making sure that each $Q_j$ is decomposed in the "same way" and that each element of its decomposition $\tilde{P}_i^j$ precisely corresponds (for what the algebra $\ca_B$ can do) to an element $\tilde{P}_i^k$ of the decomposition of another $Q_k$. Also note that the $A \in \ca_S$ that we choose is a priori different for each pair $Q_j,Q_k$ but one can prove that there exists a, this time unique, $A$ that does the job for all pairs $Q_j,Q_k$. We choose this $A$ for the rest of the proof. It remains to show that $\{\tilde{P}_i^j\}$ is a maximal family of non-zero, orthogonal and pairwise orthogonal projectors of $\ca_B$.
 
First we Observe that for all $i_1 \neq i_2$ and for all $B \in \ca_S$, 
\begin{equation}
\label{diff_i}
\begin{split}
    \tp_{i_1}^1 B \tp_{i_2}^1 & = \tp_{i_1}^1 Q_1 B Q_1 \tp_{i_2}^1 \\
    & = \alpha \tp_{i_1}^1 Q_1 \tp_{i_2}^1 \\
    & = \alpha \tp_{i_1}^1 \tp_{i_2}^1 \\
    & = 0
\end{split}    
\end{equation}

It follows that the $\tp_i^j$ are pairwise orthogonal projectors. Indeed if $j \neq k$

\begin{equation}
    \tp_{i_1}^j \tp_{i_2}^k = \tp_{i_1}^j Q_j Q_k \tp_{i_2}^k = 0
\end{equation}

by orthogonality of $Q_j$ and $Q_k$. Now suppose that $j = k$,

\begin{equation}
    \tp_{i_1}^j \tp_{i_2}^j = Q_j A^{\dagger} P_{i_1}^1 A Q_j A^{\dagger} P_{i_2}^1 A Q_j
\end{equation}

If $i_1 \neq i_2$, by \ref{diff_i} we have that $P_{i_1}^1 A Q_j A^{\dagger} P_{i_2}^1 = 0$. However in the case where $i_1 = i_2$, because the $P_i^1$ are elements of $\{P_i\}$ they are such that for all $B \in \ca_B$ there exists a $\lambda \in \mathbb{C}$ such that $P_i^1 B P_i^1 = \lambda P_i^1$; and we get that 

\begin{equation}
\begin{split}
    \tp_{i_1}^j \tp_{i_2}^j & = Q_j A^{\dagger} P_{i_1}^1 A Q_j A^{\dagger} P_{i_2}^1 A Q_j \\
    & = \lambda Q_j A^{\dagger} P_{i_1}^1 A Q_j \\
    & =  \lambda \tp_{i_1}^j
\end{split}
\end{equation}

Here, $\lambda = 1$ as

\begin{equation}
\begin{split}
    P_{i_1}^1 A Q_j A^{\dagger} P_{i_1}^1 & = P_{i_1}^1 Q_1 A Q_j A^{\dagger} Q_1 P_{i_1}^1 \\
    & = \tp_{i_1}^1 Q_1 \tp_{i_1}^1 \\
    & = \tp_{i_1}^1
\end{split}
\end{equation}

by the choice of $A$. This proves that the $\tp_{i}^j$ form a family of pairwise orthogonal projectors. The $\tp_{i}^j$ are moreover orthogonal by construction. It remains to prove that this family is maximal. Suppose that this new family is not maximal, it implies that one of the $\tp_{i}^j$ could be decomposed as $\tp_{i}^j = R_1 + R_2$ with these two new projectors being orthogonal, orthogonal to each other as well as to any $\tp_{i'}^{j'}$. This would give a more fine-grained decomposition of $Q_j$ and by applying the above reasoning in the reverse direction (transferring this time the decomposition of $Q_j$ to $Q_1$) we could refine the decomposition of $Q_1$ which is already a decomposition given by a maximal family and cannot be fine-grained! Note that with this decomposition of the $Q_j$ as a sum of $\tp_i^j$, if $Q_j$ and $Q_k$ are projectors of the same equivalence class $J$, every $\tp_i^j$ has a unique corresponding projector in the decomposition of $Q_k$: $\tp_i^k$, such that $\tp_i^j \ca_S \tp_i^k \neq 0$. Every $Q_j$ has the same number of terms in his $\tp_i^j$ decomposition and the $\{\tp_i^j,\tp_i^k,\tp_i^l,...\}$ form equivalence classes that can be represented by the index $i$.

We just proved that $\{\tp_i^j \}$ was a maximal family of non-zero, orthogonal and pairwise orthogonal projectors and thus, by Lemma \ref{change of chi}, $\chi$ can be written as 

\begin{equation}
    \chi = \sum_J \sum_{j \in J} \sum_{i = 1,...,n_j} (\ketbra{i,j}{u_{1_{ij}},d_{1_{ij}}} \otimes  \id_{\ch_{\R}^{I}})U\tp_{d_{1_{ij}}}^{u_{1_{ij}}} B \tp_i^{j}
\end{equation}

where $1_{ij}$ is the pair of up and down indices of $\tp_{d_{1_{ij}}}^{u_{1_{ij}}}$, a chosen representative of the class of $\tp_i^{j}$ for the block equivalence relation of $\{\tp_i^{j}\}$ as elements of $\ca_B$, i.e. $\tp_i^{j} R_B \tp_k^{l} \Leftrightarrow \tp_i^{j} \ca_B \tp_k^{l} \neq 0$. It shouldn't be confused with $1_J$ the index of the chosen representative of all $Q_j$ in the $J$-class of the equivalence relation $Q_j R_S Q_k \Leftrightarrow Q_j \ca_S Q_j \neq 0$. Which are both different from the equivalence relation $\tp_i^{j} R \tp_k^{l} \Leftrightarrow \tp_i^{j} \ca_S \tp_k^{l} \neq 0$ where, as we have seen, $\tp_i^{j}$ and $\tp_k^{l}$ are in the same class if and only if $j$ and $l$ are in the same class for $R_S$ and $i = k$.

We remind that $\zeta$ is of the form

\begin{equation}
    \zeta = \sum_J \sum_{j \in J} (\ketbra{j}{1_J} \otimes \id_{\ch_{\R}^{J}})VQ_{1_J} A Q_j 
\end{equation}

It remains to show that $\zeta \sqsubseteq \chi$ which we will do be constructing the isometries $\bigcirc$ and $\tikzcircle{4pt}$.

The white dot acts on the left output of $\chi$, so its input is a $\ket{i,j}$ telling us in which $\tp_i^{j}$ we are. From this information it will send the $j$ on the left, telling in which $Q_j$ we are (and this information matches the left side of $\zeta$). And it sends the $i$ on the right, telling in which class of $\tp_i^{j}$ (for $\ca_S$) we are.

Now the black dot. It acts on the right output of $\zeta$ so it acts on how elements of $\oplus_j \Im(Q_j)$ have been mapped into $\oplus_J \ch_{{\LL}_J}$. However it doesn't have in input the information that goes on the left of $\zeta$, i.e. from which $\Im(Q_j)$ it comes precisely. The difficulty is then to show that this information doesn't matter to recover how the element of $\ch$ would have been mapped on the right through $\chi$. The idea to construct the black dot is the following: we don't know the left part of $\ket{y}$ so we might as well map it to $\ket{1_J} \otimes \ket{y}$. This is isometric. Then we undo $\zeta$ and do $\chi$ which is isometric because applied on $\Im(\zeta)$. Finally as we didn't know in which $Q_j$ we were in the first place and we have made up the left part $\ket{1_J}$, the left part $\ket{i,j}$ that we get in the end is false a priori. That's why instead we're only going to keep the information of the equivalence class of $\tp_i^j$, i.e. $i$. By doing so we match the type of the right output of the white dot and and keep the whole process isometric.

It remains to show that we recover on the right the mapping that $\chi$ would have done and that the equivalence class on the left of the black dot is the same as the one on the right of the white dot.

Let $\ket{x} \in \Im(\tp_k^{p}) \subseteq \ch$. It is mapped by $\chi$ to $\ket{k,p} \otimes \ket{y}$ and mapped by $\zeta$ to $\ket{p} \otimes \ket{z}$.

\begin{equation}
    (\bigcirc \otimes \id) \chi \ket{x} = \ket{p} \otimes \ket{k} \otimes \ket{y}.
\end{equation}

\begin{equation}
    (\id \otimes \tikzcircle{4pt}) \zeta \ket {x} = \ket{p} \otimes ((\ket{i,j}\rightarrow\ket{i}) \otimes \id_{\ch_{\R}^B})W_B W_S^{\dagger} (\ket{1_{P}} \otimes \ket{z})
\end{equation}

First let's Observe that 

\begin{equation}
\begin{split}
    \zeta (Q_{1_{P}} A Q_p) \ket{x} & = [\sum_J \sum_{j \in J} (\ketbra{j}{1_J} \otimes \id_{\ch_{\R}^{J}})V Q_{1_J} A Q_j] (Q_{1_{P}} A Q_p) \ket{x} \\ 
    & = (\ketbra{1_{P}}{1_{P}} \otimes \id_{\ch_{\R}^{{P}}}) V (Q_{1_{P}} A Q_p) \ket{x} \\
    & = (\ketbra{1_{P}}{p} \otimes \id_{\ch_{\R}^{{P}}}) W_S \ket{x} \\
    & = (\ketbra{1_{P}}{p} \otimes \id_{\ch_{\R}^{{P}}}) \ket{p} \otimes \ket{z} \\
    & = \ket{1_{P}} \otimes \ket{z}
\end{split}    
\end{equation}

which means that 

\begin{equation}
\begin{split}
       &  \chi \zeta^{\dagger} (\ket{1_{P}} \otimes \ket{z}) \\
      = &   \chi \zeta^{\dagger}\zeta (Q_{1_{P}} A Q_p) \ket{x} \\
      = &   [\sum_J \sum_{j \in J} \sum_{i = 1,...,n_j} (\ketbra{i,j}{u_{1_{ij}},d_{1_{ij}}} \otimes \id_{\ch_{\R}^{J}})U\tp_{d_{1_{ij}}}^{u_{1_{ij}}} B \tp_i^{j}]
      (Q_{1_{P}} A Q_p) \ket{x} \\
     = &   \sum_{i = 1,...,n_{1_{P}}} (\ketbra{i,1_{{P}}}{u_{1_{i1_{{P}}}},d_{1_{i1_{{P}}}}} \otimes \id_{\ch_{\R}^{{P}}})U\tp_{d_{1_{i1_{{P}}}}}^{u_{1_{i1_{{P}}}}} B \tp_i^{1_{{P}}})
      (Q_{1_{P}} A Q_p) \ket{x} \\
     = &    \sum_{i = 1,...,n_{1_{P}}} (\ketbra{i,1_{{P}}}{u_{1_{i1_{{P}}}},d_{1_{i1_{{P}}}}} \otimes \id_{\ch_{\R}^{{P}}})U\tp_{d_{1_{i1_{{P}}}}}^{u_{1_{i1_{{P}}}}} B \tp_i^{1_{{P}}} A \tp_i^{p} \ket{x} \\
     = &   (\ketbra{k,1_{{P}}}{u_{1_{k1_{{P}}}},d_{1_{k1_{{P}}}}} \otimes \id_{\ch_{\R}^{{P}}})U\tp_{d_{1_{k1_{{P}}}}}^{u_{1_{k1_{{P}}}}} B \tp_k^{1_{{P}}} A \tp_k^{p} \ket{x}
      \\
    =  &   (\ketbra{k,1_{P}}{k,p} \otimes \id_{\ch_{\R}^{{P}}}) (\ketbra{k,p}{u_{1_{k1_{{P}}}},d_{1_{k1_{{P}}}}} \otimes \id_{\ch_{\R}^{{P}}})U\tp_{d_{1_{k1_{{P}}}}}^{u_{1_{k1_{{P}}}}} B \tp_k^{p} \ket{x}
      \\
    =  &    (\ketbra{k,1_{P}}{k,p} \otimes \id_{\ch_{\R}^{{P}}}) W_B \ket{x} \\
     = &    (\ketbra{k,1_{P}}{k,p} \otimes \id_{\ch_{\R}^{{P}}}) \ket{k,p} \otimes \ket{y} \\
     = &   \ket{k,1_{P}} \otimes \ket{y} 
\end{split}
\end{equation}

It follows that 

\begin{equation}
\begin{split}
    & ((\ket{i,j}\rightarrow\ket{i}) \otimes \id_{\ch_{\R}^{B}})W_B W_S^{\dagger} (\ket{1_{P}} \otimes \ket{z}) \\
    = &  ((\ket{i,j}\rightarrow\ket{i}) \otimes \id_{\ch_{\R}^{B}}) \ket{k,1_{P}} \otimes \ket{y}\\
    = & \ket{k} \otimes \ket{y}
\end{split}
\end{equation}

and finally that 

\begin{equation}
    (\id \otimes \tikzcircle{4pt}) \zeta \ket {x} = \ket{p} \otimes \ket{k} \otimes \ket{y} = (\bigcirc \otimes \id) \chi \ket{x} 
\end{equation}

which concludes the proof.
\end{proof}

\begin{remark}
Observe that the isometries $\tikzcircle{4pt}$ and $\bigcirc$ that we constructed in the previous proof are themselves canonical splitting maps such that $\cons_{\LL}(\bigcirc) = \cons_{\LL}(\zeta)$ and $\cons_{\R}(\tikzcircle{4pt}) = \cons_{\R}(\chi)$. 
\end{remark}

\begin{proof}
Indeed, let us consider the isometry $\bigcirc$. By definition its input space is $\chlc = \Span(\{ \ket{i,j} \})$ such that the $\tp_i^{j}$ are a maximal family of projectors of $\ca_B$ that decompose the $Q_j$ in a coherent way (i.e. such that $\tilde{P}_i^j =  Q_j A^{\dagger} P_i^1 A Q_j$ for all $j$ in a fixed equivalence class $J$). In particular this implies that all projectors $Q_j$ of the same equivalence class have a similar decomposition and thus that the $\ket{i,j}$ coming from these projectors are indexed by the set $I_J \times J$ where $I_J$ is the set of indices $i$ appearing in the decompositions of the $Q_j$ when $j \in J$. Then, because two projectors $Q$ not belonging to the same equivalence class have no relation between their respective $\tp$ decompositions, it follows that $\chlc = \bigoplus_J \Span(\{ \ket{i,j} | \, j \in J , \, i \in I_J \} \cong \bigoplus_J \Span( \{ \ket{i} | \, i \in I_J \}) \otimes \Span( \{ \ket{j} | \, j \in J\})$ and that $\bigcirc$ that sends $j$ to the left and $i$ to the right is a canonical splitting map. And moreover $\cons_{\LL}(\bigcirc) = \bigoplus_J \cl(\Span( \{ \ket{j} | \, j \in J\})) = \cons_{\LL}(\chi)$. 

Similarly, let us consider the isometry $\tikzcircle{4pt}$. To show that it's inducing comprehension, we proved that elements of $\ch$ that are mapped on the right in the same way by $\zeta$, are also mapped on the right in the same way by $\chi$. More precisely we've shown that $\tikzcircle{4pt} : \oplus_J \ch_{{\LL}_J} \cong \bigoplus_{J,i\in I_J} \Im(\tp_i^{j}) \rightarrow (\bigoplus_{J, i \in I_J} \Span(\{ \ket{i} \})) \bigotimes (\bigoplus_{J, i \in I_J} \ch_{\R}^{iJ})$ where $\ch_{\R}^{iJ}$ is the subset of $\chrc$ to which are sent (in a coherent way) the $(\Im(\tp_i^{j}))_{j \in J}$, through the representation unitary $U$. It follows that $\tikzcircle{4pt}$ is, by construction, a canonical splitting map and that $\cons_{\R}(\tikzcircle{4pt}) = \bigoplus_{J,i\in I_J} \cl(\ch_{\R}^{iJ}) = \cons_{\R}(\chi)$.
\end{proof}

\section{Correspondence between diagrammatic and algebraic semi-causality definitions}

We prove in this appendix, that for each of the Heisenberg and the Schrödinger semi-causality, the algebraic definition and the diagrammatic one are actually equivalent. We begin by the Heisenberg semi-causality of which we remind the definition.

\begin{definition}[Heisenberg semi-causality]
We say that a unitary $U: \ch \rightarrow \ck$ is {\em Heisenberg semi-causal} from $\mathcal{A} \subseteq \cl(\ch)$ to $\mathcal{B} \subseteq \cl(\ck)$ if 
\[  [U^{\dagger} \mathcal{B} U, \mathcal{A}] = 0 . \]
In terms of splitting maps, this is equivalent to asking that for any choice of lean splitting maps $\chi_{\ca}$, $\chi_{\ca'}$ and $\chi_{\cb}$ representing the corresponding algebras and for any $B \in \ch_{\LL}^{\chi_{\cb}}$, there exists $A_1 \in \ch_{\LL}^{\chi_{\ca'}}$ and $A_2 \in \ch_{\R}^{\chi_{\ca}}$ such that 
\begin{equation}
\input{figures/HSBtoSC_0.tikz}
\end{equation}
\end{definition}

\begin{proof}
Suppose that $[U^{\dagger} \mathcal{B} U, \mathcal{A}] = 0$. Let $\chi_{\ca}$, $\chi_{\ca'}$ and $\chi_{\cb}$ be lean splitting maps representing the corresponding algebras and let $B \in \ch_{\LL}^{\chi_{\cb}}$. Because all the splitting maps are lean and thus balanced, we know that $\loc(\chi_{\cb}) = \stloc(\chi_{\cb}) = \cb$ and similarly for $\chi_{\ca}$ and $\chi_{\ca'}$. It follows that $ \chi_{\cb}^{\dagger} (B \otimes \id) \chi_{\cb} \in \cb$ and thus that $U^{\dagger} \chi_{\cb}^{\dagger} (B \otimes \id) \chi_{\cb} U \in \ca' = \loc_{\LL}(\chi_{\ca'}) = \loc_{\R}(\chi_{\ca})$ which proves the existence of $A_1$ and $A_2$. The reverse implication is proved by the same arguments.
\end{proof}

Let us now focus on the Schrödinger semi-causality of which we remind the definitions.

\begin{definition}[Schrödinger semi-causality]
We say that a channel $\ce: \cl(\ch) \rightarrow \cl(\ck)$ is {\em Schrödinger semi-causal} from $\mathcal{A} \subseteq \mathcal{L}(\ch)$ to $\mathcal{B} \subseteq \mathcal{L}(\ck)$ if there exists a completely positive and trace-preserving (CPTP) map $\ce'$ such that \[ \Tr_{\mathcal{B'}}[\ce(\cdot)] = \ce'(\Tr_{\mathcal{A}}[\cdot]) . \] 
This condition can equivalently be phrased as the fact that for any choice of lean splitting maps $\chi_{\ca'}$ and $\chi_{\cb}$ there exists a CPTP map $\tilde{\ce}$ such that
\begin{equation}
\input{figures/causal_1.tikz}
\end{equation}
\end{definition}

\begin{proof}
Let $\ce: \cl(\ch) \rightarrow \cl(\ck)$ be a channel and suppose that there exists a CPTP map $\ce'$ such that  $\Tr_{\mathcal{B'}}[\ce(\cdot)] = \ce'(\Tr_{\mathcal{A}}[\cdot])$. Let $\chi_{\cb}$ and $\chi_{\ca'}$ be lean splitting maps for $\cb$ and $\ca'$ respectively. By Theorem \ref{equivalence of traces}, we know that there exists isometries $U$ and $V$ such that  $\Tr_{\chi_{\cb}}(\cdot) = U \Tr_{\cb'}(\cdot) U^{\dagger}$ and $\Tr_{\chi_{\ca'}}(\cdot) = V \Tr_{\ca}(\cdot) V^{\dagger}$. It follows that $\Tr_{\chi_{\cb}}(\ce(\cdot)) =  U \Tr_{\cb'}(\ce(\cdot)) U^{\dagger} = U \ce'(\Tr_{\ca}(\cdot)) U^{\dagger} = U \ce'( V^{\dagger} \Tr_{\chi_{\ca'}}(\cdot) V) U^{\dagger}$. The result doesn't immediately follow as the action of $V^{\dagger}$ is not trace-preserving. However Observe that $V V^{\dagger} \Tr_{\chi_{\ca'}}(\cdot) V V^{\dagger}$. We can thus extend $V^{\dagger}$ to a $W$ such that $W \ket{x} = V \ket{x}$ if $\ket{x} \in \Im(V)$ and $W \ket{x}  = |\ket{x}|\ket{y}$ where $\ket{y} \in \ch_{\LL}^{\chi_{\cb}}$ is of norm one. The action of $W$ is trace-preserving. And $W \Tr_{\chi_{\ca'}}(\cdot)  W^{\dagger} = W V V^{\dagger} \Tr_{\chi_{\ca'}}(\cdot) V V^{\dagger} W^{\dagger} = V^{\dagger} V V^{\dagger} \Tr_{\chi_{\ca'}}(\cdot) V V^{\dagger} V = V^{\dagger} \Tr_{\chi_{\ca'}}(\cdot) V$. Defining  $\tilde{\ce}(\cdot) = U \ce'( W \cdot W^{\dagger}) U^{\dagger}$ which is now a CPTP map makes $\ce$ diagrammatically semi-causal. A similar construction proves the reverse implication.
\end{proof}

\section{Relation between two Stinespring dilations of a map}

\begin{lemma}\label{lemma dilation}
Let $\ce : \cl(\ch) \rightarrow \cl(\ck)$ be a quantum channel. Let $U : \ch \rightarrow \ck \otimes \ck_{U}$ and $V : \ch \rightarrow \ck \otimes \ck_{V}$ be two dilations of $\ce$, i.e. isometries such that $\ce(\cdot) = \Tr_{\ck_U}(U \cdot U^{\dagger}) = \Tr_{\ck_V}(V \cdot  V^{\dagger})$, and suppose that $\dim(\ck_U) \leq \dim(\ck_V)$. Then there exists an isometry $W : \ck_U \rightarrow \ck_V$ such that $V = (\id_{\ck} \otimes W) U$.
\end{lemma}

\begin{proof}
A minimal dilation of $\ce$, is a dilation $T : \ch \rightarrow \ck \otimes \ck_T$ such that the set $\{ (A \otimes \id_{\ck_T})T\ket{\phi} , \ket{\phi} \in \ch, A \in \cl(\ck) \}$ spans $\ch \otimes \ck_T$. Consider such a dilation. By minimality of $T$, there exists \cite{Eggeling_2002} isometries $W_U : \ck_T \rightarrow \ck_U$ and $W_V : \ck_T \rightarrow \ck_V$ such that $U = (\id_{\ch} \otimes W_U)T$ and $V = (\id_{\ch} \otimes W_V)T$. Observe that, being an isometry, $W_U$ is a unitary into its image. It follows that $W_V W_U^{\dagger}$ acts isometrically from $\Im(W_U)$ (to $\Im(W_V)$) and is zero on $\Im(W_U)^{\perp}$. Then because $rank(W_U) = rank(W_V) = \dim(\ck_T) \leq \dim(\ck_U) \leq \dim(\ck_V)$, one can construct an isometry $S : \Im(W_U)^{\perp} \rightarrow \Im(W_V)^{\perp} \subseteq \ck_V$. And we can define the isometry $W = (W_V W_U^{\dagger})_{|\Im(W_U)} \oplus S : \ck_U \rightarrow \ck_V$, such that $(\id_{\ck} \otimes W)U = (\id_{\ck} \otimes W)(\id_{\ck} \otimes W_U)T = (\id_{\ck} \otimes W_V W_U^{\dagger} W_U)T = (\id_{\ck} \otimes W_V)T = V$.
\end{proof}

\end{document}